\documentclass[11pt]{article}

\newcommand{\ignore}[1]{}%

\usepackage{algorithm, algorithmic}
\usepackage{amsmath}
\usepackage{amssymb}
\usepackage{amsthm}
\usepackage{tabularx}
\usepackage{environ}
\usepackage{graphicx}
\usepackage{subcaption}
\usepackage{url}
\usepackage[margin=1in]{geometry}
\usepackage{todonotes}
\usepackage{tikz}
\usetikzlibrary{arrows,chains,shapes,matrix,positioning,calc}

\usepackage{algorithm}
\usepackage{algorithmic}
\newcommand{\INPUT}{\item[{\bf Input:}]}
\renewcommand{\algorithmiccomment}[1]{\bgroup\hfill\footnotesize~#1\egroup}

\usepackage[framemethod=tikz]{mdframed}
\mdfsetup{
	roundcorner=4pt,
	nobreak=true,
	linewidth=.5,
	innerleftmargin=25pt,
	innerrightmargin=25pt,
	innertopmargin=15pt,
	innerbottommargin=15pt,
	skipabove=12,skipbelow=8pt
}

\tikzset{
block/.style={
    rectangle,
    draw,
    text width=13em,
    text centered,
    rounded corners
},
rectangle connector/.style={
    connector,
    to path={(\tikztostart) -- ++(#1,0pt) \tikztonodes |- (\tikztotarget) },
    pos=0.5
},
rectangle connector/.default=-2cm,
straight connector/.style={
    connector,
    to path=--(\tikztotarget) \tikztonodes
}
}

\let\emptyset\varnothing

\makeatletter
\newcommand{\problemtitle}[1]{\gdef\@problemtitle{#1}}
\newcommand{\probleminput}[1]{\gdef\@probleminput{#1}}
\newcommand{\problemoutput}[1]{\gdef\@problemoutput{#1}}
\NewEnviron{problem}{
  \problemtitle{}\probleminput{}\problemoutput{}
  \BODY
  \par\addvspace{.5\baselineskip}
  \noindent
  \begin{tabularx}{\textwidth}{@{\hspace{\parindent}} l X c}
    \multicolumn{2}{@{\hspace{\parindent}}l}{\@problemtitle} \\
    \textbf{Input:} & \@probleminput \\
    \textbf{Output:} & \@problemoutput
  \end{tabularx}
  \par\addvspace{.5\baselineskip}
}

\theoremstyle{plain}
\newtheorem{theorem}{Theorem}
\newtheorem*{theorem*}{Theorem}
\newtheorem{lemma}[theorem]{Lemma}
\newtheorem{claim}[theorem]{Claim}
\newtheorem{proposition}[theorem]{Proposition}

\newtheorem{corollary}[theorem]{Corollary}

\theoremstyle{definition}
\newtheorem{definition}{Definition}

\theoremstyle{remark}
\newtheorem*{remark}{Remark}

\newcommand{\ED}{\operatorname{ED}}
\newcommand{\LCS}{\operatorname{LCS}}
\newcommand{\eps}{\epsilon}

\newcommand{\nextw}{\operatorname{next}}
\newcommand{\prevw}{\operatorname{prev}}

\newcommand{\snap}{\operatorname{snap}}

\newcommand{\Tau}{T}
\newcommand{\grho}{\rho^{\operatorname{global}}}
\newcommand{\lrho}{\rho^{\operatorname{local}}}
\newcommand{\rrho}{\rho^{\operatorname{relative}}}

\newcommand{\cA}{\mathcal{A}}
\newcommand{\cB}{\mathcal{B}}

\newcommand{\cE}{\mathcal{E}}
\newcommand{\cI}{\mathcal{L}}
\newcommand{\cJ}{\mathcal{J}}
\newcommand{\cL}{\mathcal{I}}

\newcommand{\phiw}{\phi_{\textsc{window}}}
\newcommand{\phim}{\phi_{\textsc{matching}}}
\newcommand{\shift}{\textsc{shift}}

\title{Constant-factor approximation of near-linear edit distance in near-linear time\thanks{We thank Cl{\'e}ment Canonne, Ray Li, Seri Khoury, Barna Saha, and anonymous reviewers for valuable suggestions for the manuscript.}}
\author{Joshua Brakensiek\thanks{Stanford University, supported by the NSF GRFP} \and Aviad Rubinstein\thanks{Stanford University}}
\date{}

\begin{document}

\maketitle
\thispagestyle{empty}
\begin{abstract}
We show that the edit distance between two strings of length $n$ can be computed within a factor of $f(\epsilon)$ in $n^{1+\epsilon}$ time as long as the edit distance is at least $n^{1-\delta}$ for some $\delta(\epsilon) > 0$.
\end{abstract}

\section{Introduction}\label{sec:intro}

We study the problem of computing the edit distance between two strings $A,B$ of length $n$ each.
A classical dynamic programming algorithm solves the problem exactly in quadratic time, and assuming fine-grained complexity such as the Strong Exponential Time Hypothesis (SETH), this is essentially the best running time possible for exact algorithms~\cite{BI15-Edit, AHWW16-polylog_shaved}.

There is a long sequence of approximation algorithms that run in linear or near-linear time~\cite{BJKK04-edit, BES06-edit, OR07-edit, AKO10-edit, AO12-edit}, but the state of the art is a super-constant (in particular polylogarithmic) approximation factor. 
Obtaining constant-factor approximation in truly sub-quadratic time was an outstanding open problem for a long time until the recent breakthroughs of~\cite{BEGHS18-edit-quantum} who did it in quantum sub-quadratic time, and~\cite{CDGKS18-edit} who finally achieved constant-factor approximation in (classical, randomized) truly sub-quadratic time.

Here, we build on (and significantly extend) the aforementioned sub-quadratic time approximation algorithms of~\cite{BEGHS18-edit-quantum, CDGKS18-edit} to obtain an approximation algorithm that runs in near-linear $n^{1+\eps}$ time.
However, there is a caveat: our algorithm also incurs an additive error of $n^{1-\delta}$, and hence only gives constant-factor approximation when the distance is relatively large.
Interestingly, in most settings small edit distance is actually a significantly easier problem. For example, if the distance is bounded by $\Delta$, it can be computed exactly in time $O(n\log(n) + \Delta^2)$~\cite{LMS98-n+delta^2}.

\begin{theorem}[Main result]\label{thm:main}
For any constant $\eps>0$ there exist constants $\delta \in (0, 1), c > 1$ that depend on $\eps$ (but not on $n$) such that there is a randomized algorithm that runs in time $O(n^{1+\eps})$, and given two strings $A,B$ the algorithm returns a transformation of $A$ to $B$ of cost $\leq  c\cdot \ED(A,B) + n^{1-\delta}$.
\end{theorem}

\subsubsection*{Implication for Longest Common Subsequence}
 For exact computation, the longest common subsequence (LCS) problem is equivalent to edit distance.\footnote{More precisely, LCS is equivalent to edit distance computation without substitutions since $\ED_{\textsc{no-subs}}(A,B) = 2n-\LCS(A,B)$.
 For our purposes the question of whether we allow substitutions in the definition of edit distance is irrelevant since the two definitions are equivalent up to a multiplicative factor of $2$.}
 But their multiplicative approximability is quite different, much in the same way that multiplicative approximation algorithms for Vertex Cover do not imply multiplicative approximations for Independent Set.
 For binary strings, LCS admits a trivial linear-time $2$-approximation algorithm: use only the more common symbol. 
 Obtaining better-than-$2$ approximations is a long standing open problem. 
 Concurrent work by Rubinstein and Song~\cite{RS19-LCS}, gives a fine-grained reduction from a constant-factor approximation of Edit Distance to better-than-$2$ approximation for binary LCS.
 This reduction is compatible with our algorithm since it allows for a sub-linear additive error. Combining the main result in~\cite{RS19-LCS} with our main theorem, we obtain:
 \begin{corollary}
 For any constant $\eps >0$, there exists a constant $\delta > 0$ and an algorithm that obtains a $(2-\delta)$-approximation for longest common subsequence with binary strings in $O(n^{1+\eps})$ time.
 \end{corollary}

\subsubsection*{Concurrent work by Kouck{\'y} and Saks}
In concurrent and independent work, Kouck{\'y} and Saks~\cite{KS19-ED} obtain a result comparable to our main theorem. Interestingly, while they also build on~\cite{CDGKS18-edit}, their techniques are quite different.

\subsection*{High-level technical overview}

We now give a high level informal outline of our algorithm, including comparisons to recent works of~\cite{BEGHS18-edit-quantum, CDGKS18-edit}.

\subsubsection*{Step 0: Window-compatible matching}
The first step in our algorithm is to partition strings $A,B$ into  $t\approx \sqrt{n}$ windows of roughly $d \approx \sqrt{n}$ characters each. By~\cite{BEGHS18-edit-quantum}, there is a near-optimal matching of the strings that is {\em window-compatible}; i.e. all the characters from a window $a$ of string $A$ are matched to the same window $b$ of string $B$. Here ``near-optimal'' hides a small constant multiplicative factor, but also a (sub-linear) additive factor, which is the main reason that we incur this error term in our main result.

Our goal is to approximate enough of the pairwise distances between windows to reconstruct a near-optimal matching. Once we know all the pairwise distances, we can reconstruct a near-optimal window-compatible matching in time $t^2$ using the classical dynamic programming algorithm.

To approximate the distance between any given pair, we implement a ``query'' by recursing on our main theorem. This runs in time $d^{1+\eps}$, so as long as the number of queries is $t^{1+\eps}$, the total running time is $t^{1+\eps} \cdot d^{1+\eps} \approx n^{1+\eps}$.

Let $\cA, \cB$ denote the respective sets of windows in strings $A,B$. For the rest of the analysis, we consider the graph $M_{\Delta} = (\underbrace{\cA \cup \cB}_{=V_{\Delta}}, E_{\Delta})$, where we have an edge between two windows if and only if their edit distance is $\lesssim \Delta$. Below, we drop the $\Delta$ subscripts when clear from context. Notice also that $t \approx |V|$, so our goal is to ``approximate'' $M_{\Delta}$ by finding a set $M \subset M_{O(\Delta)}$ (inclusion is with respect to edges) such that $M_{\Delta} \setminus M$ is small  while making only $O(|V|^{1+\eps})$ queries.

\begin{figure*}[t!]
  \centering
  \begin{subfigure}{0.9\textwidth}
    \centering
    \includegraphics[width=6in]{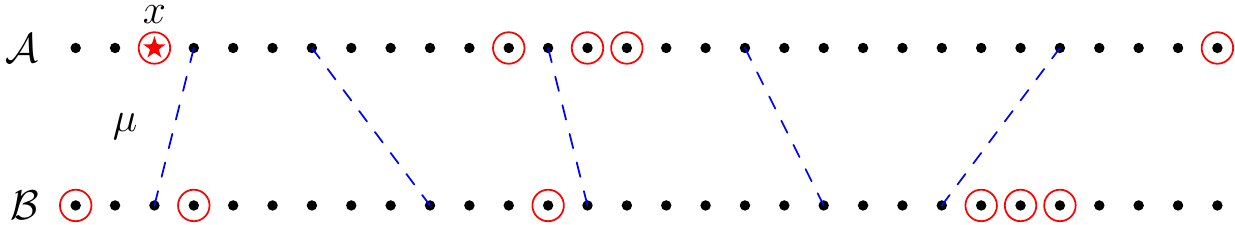}
    \caption{Step 1.}
    \label{fig:big-clique}
  \end{subfigure}\\
  \vspace{.5in}
  \centering
  \begin{subfigure}{0.9\textwidth}
    \centering
    \includegraphics[width=6in]{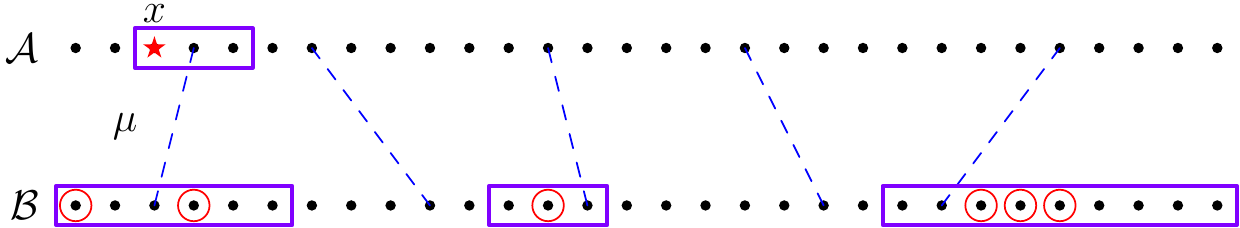}
    \caption{Step 2.}
    \label{fig:step-2-clique}
  \end{subfigure}\\
  \vspace{.5in}
  \centering
  \begin{subfigure}{0.9\textwidth}
    \centering
    \includegraphics[width=6in]{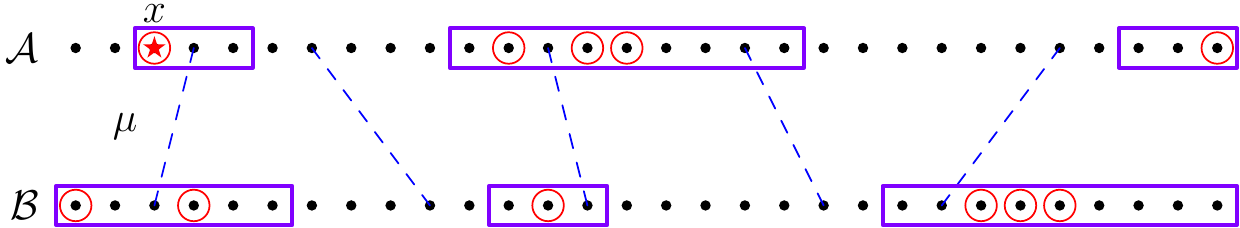}
    \caption{Step 3.}
    \label{fig:small-clique}
  \end{subfigure}
  \caption{An illustration of Steps 1-3 of the query algorithm. Each dot represents a window (Step 0), with $\cA$ be the windows of $A$ and $\cB$ the windows of $B$. The {\color{blue} dashed} edges are the partial matching $\mu$. $x$ (with a {\color{red} star}) represents the center of the clique. The {\color{red} circled} windows are included in the clique. In Figures (b) and (c), the {\color{violet} boxed} windows are kept in the next round of the algorithm tree.}
  \label{fig:query-steps}
\end{figure*}

\subsubsection*{Step 1: Dense graphs}
If the graph $M_{\Delta}$ is dense, we can query the edit distance from one window $x$ to all other windows. 
By the triangle inequality, any pair of neighbors of $x$ are also close, so the number of edges we discover is quadratic in the degree of $x$.
In total, in this approach we expect to pay roughly $|V|$ queries to discover $\deg^2(x) \approx (|E|/|V|)^2$ edges. In order to discover all $|E|$ edges in the graph, we can hope to make roughly $|V|^3/|E|$ queries.
(Actually doing it requires some care so that we are not repetitively discovering the same edges.)
See Figure~\ref{fig:big-clique} for an illustration.

In previous work, Steps 0 and 1 were the core of~\cite{BEGHS18-edit-quantum} (which resorts to a quantum algorithm for sparse graphs).

\subsubsection*{Step 2: Sparse graphs}
The key idea to deal with sparse graphs is the following simple structural observation about near-optimal matchings: we can assume without loss of generality that any pair $a_1, a_2$ of $A$-windows that are $k$-close (in terms of their position in $A$) are matched to $B$-windows $\mu(a_1), \mu(a_2)$ that are $O(k)$-close.
This is indeed without loss of generality since even if an optimal matching $\nu$ were to match $a_1, a_2$ to a pair $\nu(a_1), \nu(a_2)$ that is $\omega(k)$-far, most of the windows between $\nu(a_1), \nu(a_2)$ cannot be matched to any $A$-windows and must be deleted. Hence the cost of matching $a_1, a_2$ to the wrong windows (or deleting them) is negligible compared to the cost of deleting the spurious windows between $\nu(a_1), \nu(a_2)$.
For the rest of the analysis, we fix such a near-optimal matching $\mu$ that satisfies the above condition.

The above structural observation leads to the following ``seed-and-expand'' algorithmic approach for sparse graphs: first, query the edit distance of some $A$-window $a$ to every window; if the graph is sparse, those $O(|V|)$ queries should generate a short list of $\deg(a) \approx |E|/|V|$ candidate matches $\{b_1, \dots, b_{\deg(a)}\} \ni \mu(a)$.
Now consider the windows in an interval%
\footnote{By {\em interval} we refer to a contiguous substring of length varying from one window to the entire string.}   
$I \subset \cA$ around $a$:
we expect them to be matched to an interval $J_j$ of length $|J_j| = O(|I|)$ around one of the candidates $b_j$.
For those windows in $I$, we reduced the number of queries we need to make to roughly $|I||E|/|V|$, which is much less than the naive $|V|$ when $|I||E| \ll |V|^2$.
(Formalizing this argument requires some care since it is possible that the original $a$ is actually deleted in an optimal matching.) See Figure~\ref{fig:step-2-clique} for an illustration.
Steps 0,1,2 are the core of~\cite{CDGKS18-edit}.

Careful optimization of this step, including recursively decreasing the length of the interval $I$, can reduce the total query complexity to approximately $|E|$.
It is not clear how could one obtain better query complexity from this seed-and-expand approach: even the window immediately adjacent to $a$, has $\deg(a) \approx |E|/|V|$ candidate neighbors.

\subsubsection*{Step 3: Cliques}
As outlined above, optimizing Steps 0,1,2 can give query approximately $\min\{|V|^3/|E|, |E|\} \leq |V|^{1.5}$.
There is a tight example where the graph has $\sqrt{|V|}$ cliques of size $\sqrt{|V|}$, hence using either the sparse or dense graph approaches is stuck at $\approx |V|^{1.5}$ queries.
When recursively applying the $\approx |V|^{1.5}$ queries algorithm, one can obtain an approximation algorithm for edit distance with run-time $\approx n^{1.5}$.

Our key novel idea for this case is to combine both approaches: we query the edit distance from one window $y \in V$ to all other windows using $O(|V|)$ queries. Thus, as in Step 1, we discover one clique of $\sqrt{|V|}$ windows. We now mimic the algorithm from Step 2, but using an interval around each one of the $\sqrt{|V|}$ windows in the clique (instead of just one interval on the A-side).
This insight allows us to make roughly $\deg(v) \approx |E|/|V|$ seed-and-expands for the price of one; this is exactly the factor that we need to reduce the $|E|$ complexity from Step 2 to $|V|$. See Figure~\ref{fig:small-clique} for an illustration.

\subsection*{Deeper technical overview}
Below we highlight in greater detail some of the ideas that go into actually  implementing and analyzing the above blueprint.

\subsubsection*{Graphs that are not disjoint unions of cliques}

One obvious obstacle is that in general the graph may actually not partition into disjoint cliques.
In more detail, suppose that we query the distance between a window $a$ and all other windows, and find $\sqrt{|V|}$ $B$-neighbors $b_1,\dots,b_{\sqrt{|V|}}$, and $\sqrt{|V|}$ $A$-neighbors $a_1,\dots,a_{\sqrt{|V|}}$. Now, we want to say that since $a_i$ is close (in edit distance) to all $b_j$, windows at an interval $I_i$ around $a_i$ are likely to be matched to windows at an interval $J_j$ around some $b_j$. 
But it is entirely possible that $a_i$ has other neighbors that are not neighbors of the original $a$. If we only look to match the windows in $I_i$ with windows in $\bigcup_j J_j$, we might miss their optimal neighbors.

To force our graphs to ``behave like'' disjoint unions of cliques, we consider a ball around window $a$ of edit distance radius $\tau\cdot\Delta$, for a randomly chosen $\tau$.
Now if it is optimal to match $a_i$ to $\mu(a_i)$ such that $\ED(a_i,\mu(a_i)) \leq \Delta$, then we have that for every $a$ and almost any $\tau$, either both $a_i,\mu(a_i)$ are in the edit distance ball of radius $\tau \cdot \Delta$ around $a$, or neither is in the ball.
More generally, in multiple parts of the proof, we use the fact that by triangle inequality, the edit-distance ball of radius $\tau\cdot\Delta$ around $a_i$ is contained in the edit-distance ball of radius $(\tau+1)\cdot\Delta$ around $\mu(a_i)$.

We henceforth continue to loosely refer to the ball around $a$ as a ``clique''; even though it may not be fully connected, it approximates a clique in the sense that every pair is $2\tau \cdot \Delta$-close in edit distance.

\subsubsection*{Query algorithm tree}
We analyze our query algorithm by considering a tree of recursive calls.
At the root, all windows are alive.
Each edge of the algorithm tree corresponds (roughly) to the following subroutine: choose a clique (aka edit distance ball around a random live window) and keep the intervals around the clique-windows.
We want to show that:
\begin{itemize}
\item On each edge of the algorithm tree (a.k.a. each run of our subroutine), we decrease the number of live windows by a polynomial ($|V|^{\eps}$) factor. Thus after a constant depth recursion we are left with a small number of live windows on which we can brute force query all the pairs.
\item In each node, a non-negligible fraction of live windows $y$ have their match $\mu(y)$ also live; we call these windows {\em good}. It is important that the fraction of good windows is non-negligible among live windows --- otherwise the algorithm has no hope finding their matches by considering the edges of nearby bad (not good) windows.
\item In each node of the algorithm tree, most good windows have $\approx|V|^{-\eps}$ probability of surviving (as live and good) to a child of that node. By sampling $\approx |V|^{\eps}$ children for each internal node, we can argue that most good windows are likely to survive to a leaf. Furthermore, notice that the total number of live windows in each layer remains roughly linear.
\item The query complexity at each node is roughly proportional to the number of live windows. Thus in total the query complexity in each level of the algorithm tree is approximately linear.
\end{itemize}

\subsubsection*{Cliques of different sizes}
Another obvious gap between the ideal example described in Step 3 and worst case instances is that even if the graph can be partitioned into disjoint cliques, they may have very different sizes.
And even if the cliques have the same size, they may be denser in some areas of the string and sparser in others.
Above we informally describe taking an interval around each clique window and keeping the windows in this interval alive for the next level.
$$\text{How large of an interval should we take around each clique window?}$$
We need to carefully balance between (i) reducing the number of live windows; (ii) making sure that good windows continue (with-not-too-small probability) to be alive in the next level; and (iii) making sure that they also continue to be good in the next level (specifically, we want the survival of $y$ to be highly correlated with the survival of $\mu(y)$). 
Instead of picking a uniform interval length, we take, for each clique-window $y$, the maximal interval $I \ni y$ such that the clique $C$ is $|V|^{\eps}$-factor denser on $I$ than on the entire string.
This ensures that the right proportion ($|V|^{-\eps}$-fraction) of live windows are covered by dense intervals and survive to the next child of the algorithm tree.

\subsubsection*{Colors}
Suppose, that a window $y$ is part of a small clique (i.e. there are few other windows that are close to $y$ in edit distance), but most windows in an interval around $y$ belong to large cliques. 
On one hand, it is unlikely that the algorithm ever samples $y$'s small clique (because it is small). On the other hand, even if it did sample a clique with a window $z$ that is relatively close (on the string) to $y$, that clique is large, and hence we can only take a very small interval around $z$ --- too small to contain $y$.
In this case, we may actually lose $y$. This is another source of our additive approximation error (in addition to the loss from the window-compatible matching in Step 0). 

In order to bound this error term, we partition the windows into a constant number of {\em colors}, or equivalence classes based on the statistics of their cliques. By a simple Markov argument, we show that most windows $y$ have a not-too-small fraction of same-color windows in any interval around them. Whenever this is the case, the probability of sampling any clique in that interval is proportional to the length of the interval that the algorithm would use if it sampled $y$'s clique. Thus such a $y$ is likely to be discovered by the algorithm and survive to the next level.

\subsubsection*{On additive error}
Our algorithm incurs an additive error of $n^{1-\exp(-1/\eps)}$ to achieve a runtime of $n^{1+\eps}$. The reason is that the our algorithm makes $O(1/\eps)$ levels of recursion. In each layer of the recursion, the size of the string is square-rooted. At the bottom-most level, when our string is of size $n' = n^{\exp(-1/\eps)}$ we incur a $n'^{1-c}$ additive error for some $c > 0$. When the $\approx n/n'$ subtrees are combined, this gives an additve error of $n/(n')^c \approx n^{1-\exp(-1/\eps)}$.

\subsection{Succinct Roadmap}

In Section~\ref{sec:prelim}, we precisely formulate the edit distance problem and the parameters needed in the proof. In Section~\ref{sec:windows}, we formulate how to compute window edit distance  and reduce approximating edit distance to only needing to consider low-skew mappings. In Section~\ref{sec:alg}, we describe the reduction from edit distance to a query algorithm. In Section~\ref{sec:query}, we describe the query algorithm as well as analyze it.

As the algorithm and analysis both have many moving parts, we have provided flow charts of the algorithms (Figure~\ref{fig:algflow}) and key lemmas and theorems (Figure~\ref{fig:lemmaflow}) in this paper.

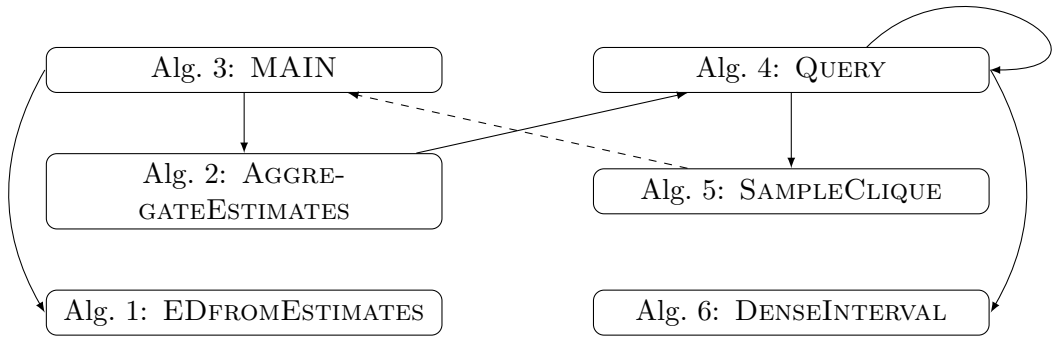
\begin{figure}
  \begin{center}

\begin{tikzpicture}
  \matrix (m)[matrix of nodes, column  sep=2cm,row  sep=8mm, align=center, nodes={rectangle,draw, anchor=center} ]{
    |[block]| {Alg.~\ref{alg:alg2query}: \textsc{MAIN}} & |[block]| {Alg.~\ref{alg:query-recursion}: \textsc{Query}}\\
    |[block]| {Alg.~\ref{alg:estimate}: \textsc{AggregateEstimates}} & |[block]| {Alg.~\ref{alg:5}: \textsc{SampleClique}}\\
   |[block]| {Alg.~\ref{alg:edfromest}: \textsc{EDfromEstimates}} & |[block]| {Alg.~\ref{alg:6}: \textsc{DenseInterval}}  \\
};
\path [>=latex,->] (m-1-1) edge (m-2-1);
\draw [>=latex,->] ($(m-1-1)-(2.65,0)$) to [bend right] ($(m-3-1)+(-2.65,0)$);
\path [>=latex,->] ($(m-1-2)+(1.0,0.3)$) edge [out=45,in=360,looseness=3] (m-1-2);
\path [>=latex,->] (m-2-1) edge (m-1-2);
\path [>=latex,->] (m-1-2) edge (m-2-2);
\path [>=latex,->] (m-2-2) edge [dashed] (m-1-1);
\draw [>=latex,->] ($(m-1-2)+(2.65,0)$) to [bend left]  ($(m-3-2)+(2.65,0)$);
\end{tikzpicture}

\end{center}
\caption{Flow chart of algorithms. An arrow from one algorithm to another indicates that the former algorithm calls the latter as a subroutine. The first column indicates the algorithms analyzed in Section~\ref{sec:alg} and the second column indicates the algorithms analyzed in Section~\ref{sec:query}. Note that \textsc{Query} calls itself. The dashed arrow from \textsc{SimpleClique} to \textsc{MAIN} indicates that the parameter $L$ decreases by $1$ when \textsc{SimpleClique} calls \textsc{MAIN}.}
\label{fig:algflow}
\end{figure}

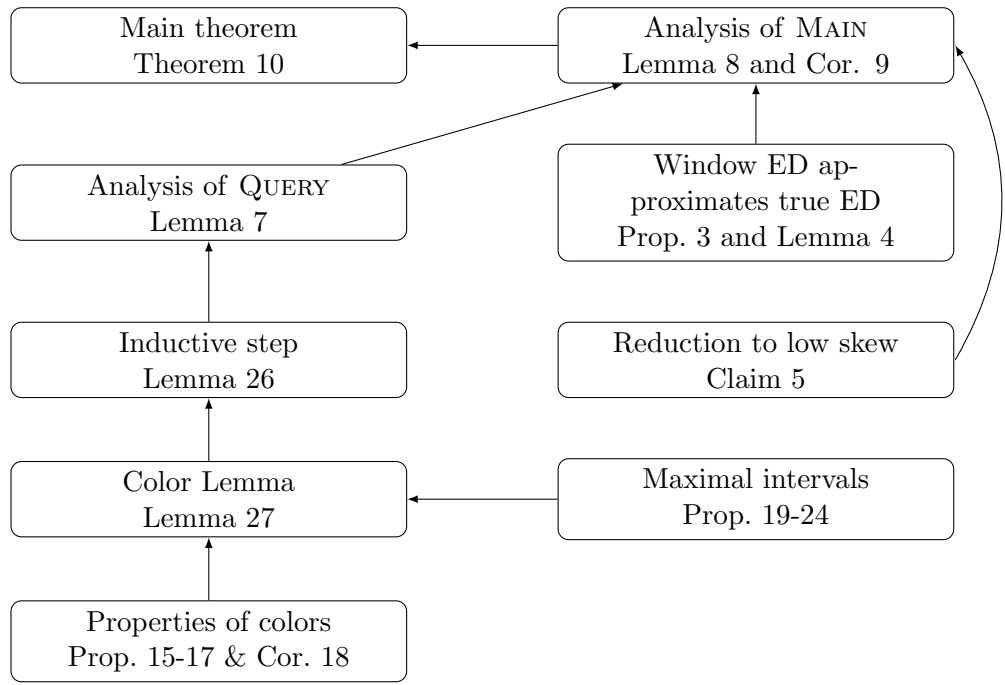
\begin{figure}
  \begin{center}

\begin{tikzpicture}
  \matrix (m)[matrix of nodes, column  sep=2cm,row  sep=8mm, align=center, nodes={rectangle,draw, anchor=center} ]{
  |[block]| {Main theorem\\Theorem 10} & |[block]| {Analysis of \textsc{Main}\\ Lemma 8 and Cor. 9} \\
   |[block]| {Analysis of \textsc{Query}\\Lemma 7} &  |[block]| {Window ED approximates true ED\\Prop.~\ref{prop:ED-mu-lower-bound} and Lemma~\ref{prop:ED-mu-upper-bound}}\\
   |[block]| {Inductive step\\Lemma 26}  & |[block]| {Reduction to low skew\\Claim~\ref{claim:low-skew}}\\
   |[block]| {Color Lemma\\Lemma~27} & |[block]| {Maximal intervals\\
   Prop.~19-24}  \\
   |[block]| {Properties of colors\\Prop.~15-17 \& Cor.~18}\\ 
};
\path [>=latex,->] (m-1-2) edge (m-1-1);
\draw [>=latex,->] (m-4-1) edge (m-3-1);
\path [>=latex,->] (m-2-1) edge (m-1-2);
\path [>=latex,->] (m-3-1) edge (m-2-1);
\path [>=latex,->] (m-2-2) edge (m-1-2);
\draw [>=latex,->] ($(m-3-2)+(2.65,0)$) to [bend right]  ($(m-1-2)+(2.65,0)$) ;
\draw [>=latex,->] (m-4-2) edge (m-4-1);
\draw [>=latex,->] (m-5-1) edge (m-4-1);
\end{tikzpicture}

  \end{center}
  \caption{Flow chart of the lemmas and theorems in the paper. An arrow from one group of results to another demostrates that the former results were used to prove the latter.}
  \label{fig:lemmaflow}
\end{figure}

\section{Preliminaries}\label{sec:prelim}

\subsection{Problem Definition}\label{subsec:probdef}

We let $[n] := \{1, \hdots, n\}$ denote the index set of our strings. We denote the alphabet of our strings by $\Sigma$. We let $\bot \not\in\Sigma$ denote an additional character.

In this paper, strings are $1$-indexed. If $A \in \Sigma^n$ is a string, then we let $A[i, j]$, denote the substring of characters with index between $i$ and $j$, inclusive, where $1 \le i \le j \le n$.

We define the \emph{edit distance} between two strings $A, B \in \Sigma^n$ to be the minimum number of insertions, deletions, and substitutions needed to transform $A$ into $B$. We note this quantity as $\ED(A, B)$.

\textbf{Simplifying assumptions.} As our goal is to find an algorithm when $c$ is a constant (but not any particular constant), it suffices to consider the case where the strings have the same length. Otherwise, since  $\ED(A, B) \ge |n_A - n_B|$, we can delete the last $|n_A-n_B|$ characters of the longer of $A$ and $B$ to reduce to the case the two strings have the same length without precluding a constant-approximation.

We also assume without loss of generality that all necessary expressions are integral. Rounding has negligible impact on the approximations.

\paragraph{Main theorem}
We now restate our Theorem~\ref{thm:main} more formally:

\begin{theorem*}
  For all $\delta > 0$, there exists constants $\alpha_{\delta} \ge 1$ and $\eps_{\delta} > 0$ an $O(n^{1+\delta})$ randomized algorithm which takes as input two strings $A$ and $B$ of length $n$ and outputs $\widehat{ED}(A, B) \ge \ED(A, B)$ such that with probability at least $\frac{2}{3}$,
  \[
    \widehat{ED}(A, B) \le \alpha_{\delta}\cdot\ED(A, B) + n^{1-\eps_{\delta}},
  \]
  where $\alpha_{\delta} \approx \exp(\exp(1/\delta \cdot \log(1/\delta)))$ and $\eps_{\delta} \approx \exp(-O(1/\delta))$.
\end{theorem*}

\begin{remark}
  By standard amplification arguments, the probability of failure can be made subexponential in $n$.
\end{remark}

\begin{remark}
  Our algorithm can also efficiently output with probability at least $\frac{2}{3}$ a sequence of edits between $A$ and $B$ of length at most $\alpha_{\delta}\cdot \ED(A, B) + n^{1-\eps_{\delta}}$. This follows from~\cite{CGKK18-ED}.
\end{remark}

\subsection{Table of parameters}\label{subsec:parameters}

As the proof involves many terms and parameters, we list them in Table~\ref{table:param} for reference. Each definition is restated when defined in the proof.

\begin{table}[!htbp]
  \begin{center}
 \caption{Table of parameters}
  \label{table:param}
    \begin{tabular}{rcl}
      Term & Definition & Note\\\hline
      $\Sigma$ & Alphabet & \\
      $n$ & Input length & \\
      $A, B$ & Input strings & \\
      $\ED(A, B)$ & Edit distance & \\\\

      $\delta$ & $n^{1+\delta}$ is target running time &\\
      $L_{\max}$ & $\log_2(2/\delta)$ & Max level of $\textsc{MAIN}$\\
      $L$ & $L \in \{0, 1, \hdots, L_{\max}\}$ & Current level of $\textsc{MAIN}$\\
      $i_{\max}$ & $10/\delta$ & Max level of $\textsc{Query}$.\\
      $i$ & $i \in \{0, 1, \hdots, i_{\max}\}$ & Current level of $\textsc{Query}$.\\\\

      $\eps$ & $1/200^{L_{\max} + i_{\max}+1}$ & $n^{1-\eps}$ is bound on additive error in Lemma~\ref{lem:main-query}\\
      $\eps'$ & $1/200^{L_{\max} + i_{\max}+2}$ & $n^{1-\eps'}$ is bound on additive error in Lemma~\ref{lem:main-estimate}\\
      $t_{min}$ & $(1000/\eps^{10})^{4^{L+1}/\eps^2}$ & lower bound on $t$\\
      $n_{min}$ & $(1000/(\eps')^{10})^{4^{L+2}/(\eps')^2}$ & lower bound on $n$\\
      $\Delta$ & Edit distance additive error & $\Delta \ge n^{1-\eps}$\\
      $d$ & $\sqrt{n}$ & Window width\\
      $\gamma$ & $\frac{\Delta d}{n}$ & Window spacing, $\ge n^{1/2 - \eps}$\\
      $\cA$ & $\{A[1, d], A[d+1, 2d], \hdots, A[n-d+1, n]\}$ & windows of $A$\\\\
      $\cB$ & $\{B[1, d], B[\gamma+1, \gamma+d], \hdots, B[n-d+1, n]\}$ & Windows of $B$\\
      $t$ & $|\cA| + |\cB|$ & Number of windows, $\le n^{1/2 + 2\eps}$.\\
\\

      $\alpha_L$ & $2^{(20000/\eps^2)^L}$ & APX factor of \textsc{MAIN} recursion level $L$\\
      $c_L$ & $100\alpha_{L-1}$ & metric ball expansion rate\\
      $\beta_L$ & $(c_L)^{\tau_{\max} + 2}$ & upper bound on APX factor of \textsc{QUERY}\\
      $\eps_i$ & $100^{i+1}\eps$ & \\
      $\snap(\ell)$ & $t^{\lfloor \eps \log_t(\ell)\rfloor/\eps}$ & Round $\ell$ to power of $t^{\eps}$\\
      $\cL$ & intervals of $\cA$ and $\cB$ of length $t^0,t^{\epsilon},\dots,t$ & for each length, the intervals are\\ &&disjoint and cover $\cA \cup \cB$.\\\\

      $\tau_{\max}$ & $1000/\eps^3$ & upper bound on distance threshold\\
      $T$ & $\{1, \dots, \tau_{\max}\}$ & distance thresholds\\
      $\tau$ & $\in_{u.a.r} T$ & random threshold for distance\\
      $\rho$ & $t^{1/i_{\max}}$ & relative density threshold\\
      $\Lambda$ & $\{1, 7\}$ & set of interval multipliers\\
      $\lambda$ & $\in \Lambda$ & interval multiplier
    \end{tabular}
  \end{center}
\end{table}

\section{Mappings, Windows and Skew}\label{sec:windows}

In this section, we define important concepts for studying the edit distance between strings.

\subsection{Mappings}

A \emph{mapping} is a partial function $\mu : A \to B \cup \{\perp\}$. This mapping is \emph{monotone} if all for all $i < i'$ such that $\mu(A_i) = B_j$ and $\mu(A_{i'}) = B_{j'}$, then $j < j'$. We abuse notation and say that $A_i \in \mu$ is $\mu(A_i) \neq \perp$. Observe that the insertion-deletion distance between $A$ and $B$ is equal to $2n - 2|\mu|$, for the maximal choice of $\mu$.

\subsection{Windows}

Let $\Delta \in [n^{1-\eps}, n]$ be the target additive error in our edit distance computation. We seek to divide the strings $A$ and $B$ into \emph{windows}, or contiguous substrings, of length $d := \sqrt{n}$. For string $B$, we also include windows that correspond to shifts of integer multiples of $\gamma := \frac{\Delta d}{n}$.

Let $t := 1 + \frac{n}{d}+\frac{n-d}{\gamma}$.

\begin{definition}[Windows]
  We partition strings $A,B$ into total $t$ overlapping {\em windows}, or contiguous substrings of width $d$. Concretely,
  \begin{align*}
    \cA &:= \{A[1, d], A[d+1, 2d], \hdots, A[n - d + 1, n]\}.\\
    \cB &:= \{B[1, d], B[\gamma + 1, \gamma + d], \hdots, B[n - d + 1, n]\}.
  \end{align*}
  Notice that $\cA$ has a spacing of $d$ and $\cB$ has a spacing of $\gamma$.
\end{definition}

For window $a \in \cA$, let $s(a)$ denote the starting index of $a$ (e.g., $s(A[1, d]) = 1$).

\subsubsection*{Mappings between window sets}

Generalizing how we defined a mapping between characters in a string, we say that a mapping $\mu : \cA \to \cB \cup \{\bot\}$ between windows is \emph{monotone} if for all $a, a' \in \cA$ such that $\mu(a), \mu(a) \neq \bot$ and $s(a) \le s(a')$ then $s(\mu(a)) \le s(\mu(a'))$. Setting $\mu(a) = \bot$ represents deleting $a$ from the string. As such, we define $\ED(a, \bot) = d$ for all windows $a$.

As a minor abuse of notation, we let $\mu \subseteq \cA$ denote the set of $A$-windows such that $\mu(a) \neq \perp$.

For $a \in \mu$, let $a.\nextw$ denote the $a' \in \mu$ immediately after $a$ (note that $\nextw$ depends on the mapping $\mu$). If $a$ is the last window in $\mu$, we define $a.\nextw := \perp$. We define $a.\prevw$ in the analogous way.

We say that the {\bf edit distance of mapping $\mu$} is:
\[\ED(\mu) := 2\underbrace{\sum_{a\in \cA} \ED(a,\mu(a))}_{(*)}
+ 2\sum_{\substack{a \in \cA, \mu(a)\neq \perp\\a.\nextw \neq \bot}} |s(\mu(a))+d-s(\mu(a.\nextw))|.\]
The first term is just sum of the edit distances, and the second term is a penalty for either overlap (requiring deletions) or excessive spacing (requiring insertions).

We now show that the minimal value of $\ED(\mu)$ over all monotone $\mu : \cA \to \cB \cup \{\perp\}$ is a good approximation of $\ED(A, B)$, up to certain additive and multiplicative factors.
First, we show that $\ED(\mu)$ cannot underestimate $\ED(A, B)$.

\begin{proposition}[implicit in \cite{BEGHS18-edit-quantum}]\label{prop:ED-mu-lower-bound}
  For all monotone mappings $\mu : \cA \to \cB \cup\{\perp\}$, we have that $\ED(\mu) \ge \ED(A, B).$
\end{proposition}

\begin{proof}
  Using $\mu$ we construct an explicit mapping from $A$ to $B$. For each window $a \in \cA$, transform the characters of $a$ in $A$ into $\mu(a)$. If $\mu(a) = \perp$, then delete all the characters of $a$. For any remaining window $a$ (now $\mu(a)$) which is not last, if $s(\mu(a)) + d - s(\mu(a.\nextw)) > 0$, then delete that many characters from the end of $\mu(a)$.

  Let $\mu(A)$ be the currently transformed string. By construction,
  \[
    \ED(A, \mu(A)) \le \sum_{a\in \cA} \ED(a,\mu(a))
+ \sum_{\substack{a \in \cA, \mu(a)\neq \perp\\a.\nextw \neq \bot}} |s(\mu(a))+d-s(\mu(a.\nextw))|.\]
  Furthermore, since we deleted any overlaps between $\mu(a)$'s, we have that $\mu(A)$ is a subsequence of $B$. Thus,
  \[
    \ED(\mu(A), B) = |B| - |\mu(A)| = |A| - |\mu(A)| \le \ED(A,\mu(A)).
  \]
  Thus,
  \begin{align*}
    \ED(A, B) &\le \ED(A, \mu(A)) + \ED(\mu(A), B)\\
              &\le 2\sum_{a\in \cA} \ED(a,\mu(a))
                + 2\sum_{\substack{a \in \cA, \mu(a)\neq \perp\\a.\nextw \neq \bot}} |s(\mu(a))+d-s(\mu(a.\nextw))|\\
              &= \ED(\mu).
  \end{align*}
\end{proof}

Next, we show that a ``good'' mapping $\mu$ with $\ED(\mu) \approx \ED(A, B)$ exists.

\begin{lemma}[variant of Lemma~4.3 of \cite{BEGHS18-edit-quantum}]\label{prop:ED-mu-upper-bound}
For all $A, B$, there exists a monotone mapping $\mu_{A\rightarrow B} : \cA \rightarrow \cB \cup \{\perp\}$ such that $\ED(\mu) \leq 16\ED(A,B) + 12\frac{n}{d}\cdot \gamma + 4d$.
\end{lemma}

\begin{proof}
  Consider the optimal sequence of edits from $A$ to $B$. This can be viewed as $\ell$ substitutions of characters of $A$, $k$ deletions of characters, and then $k$ insertions, where $2k + \ell = \ED(A, B)$. Let $A'$ be the subsequence of untouched characters of $A$. Let $B'$ be the corresponding subsequence of $B$. Let $\mu' : A' \to B'$ be the monotone correspondence between the characters of these substrings.

  We construct $\mu_{A \to B} : \cA \to \cB \cup \{\perp\}$ as follows ($\mu$ for brevity). For each $a \in \cA$, if $a \cap A' = \emptyset$, then let $\mu(a) = \perp$. Otherwise, consider the first $A[i] \in a \cap A'$. Set $\mu(a)$ to be the rightmost intervals of $B$ which contains $\mu'(A[i])$. Since $\mu'$ is a monotone map, we have that $\mu$ is also monotone. 

  For each window $a \in \cA$ which nontrivially intersects $A'$, let $i_a$ be the first index of $\mu'(a \cap A')$. Likewise, let $j_a$ be the last index of $\mu'(a \cap A')$. Note that $i_{a.\nextw} \ge j_a$ for all $a$ with $\mu(a) \neq \perp.$ 
  
  Further define $k_a$ to be the number of characters in window $a$ which are deleted, $\ell_a$ to be the number of characters which are substituted, and $m_a$ the number of characters of $a \cap A'$ which $\mu'$ fails to map to $\mu(a)$. Note that for all $a \in \cA$
  \[
    \ED(a, \mu(a)) \le 2(k_a + \ell_a + m_a),
  \]
  as one can delete the $k_a + \ell_a + m_a$ unmatched characters and then insert the correct ones.   If $\mu(a) \neq \perp$, and $\mu(a)$ is not the $[n-d+1, n]$ block of $\cB$, we have that $0 \le i_a - s(\mu(a)) \le \gamma.$ This in particular means that
  \[
    |s(\mu(a)) + d - s(\mu(a.\nextw))| \le 2\gamma + |i_a + d - i_{a.\nextw}|
  \]
  as long as $\mu(a.\nextw) \neq [n-d+1, n]$. If $\mu(a) \neq [n-d+1, n]$ and $\mu(a.\nextw) = [n-d+1, n]$, then we can bound $|s(\mu(a)) + d - s(\mu(a.\nextw))| \le \gamma + d+ |i_a + d - i_{a.\nextw}|$. Finally, if $\mu(a) = [n-d+1,n]$, then $|s(\mu(a)) + d - s(\mu(a.\nextw))| \le d$.

  Putting all these together, we can bound
  \begin{align*}
    \ED(\mu) &= 2\sum_{a\in \cA} \ED(a,\mu(a))
               + 2\sum_{\substack{a \in \cA, \mu(a)\neq \perp\\a\text{ is not last}}} |s(\mu(a))+d-s(\mu(a.\nextw))|\\
             &\le 2\sum_{a \in \cA} 2(k_a + \ell_a + m_a) + 2\sum_{\substack{a \in \cA, \mu(a)\neq \perp\\a\text{ is not last}}} (2\gamma + |i_a + d - i_{a.\nextw}|) + d|\{a : \mu(a) = [n-d+1, n]\}|\\
             &\le 4k + 4\ell + 4\sum_{a \in \cA} m_a + 4\frac{n}{d}\gamma + 2\sum_{\substack{a \in \cA, \mu(a)\neq \perp\\a\text{ is not last}}} |i_a + d - i_{a.\nextw}| + d|\{a : \mu(a) = [n-d+1, n]\}|
  \end{align*}
  We now bound each of these terms. Clearly $2k + 2\ell \le 2\ED(A, B)$.

  Note that if $\mu(a) = \perp$ or $[n - d+1, n]$ then $m_a = 0$. Otherwise, $m_a \le j_a - s(a) - d \le j_a - i_a - d + \gamma$. Note that there must be at least $j_a - i_a - d$ insertions between $i_a$ and $j_a$ in $B$, so
  \[
    4\sum_{a \in \cA} m_a \le 4\ED(A, B) + 4\frac{n}{d}\gamma.
  \]

  For $\sum |i_a + d - i_{a.\nextw}|$, note that at least $i_a + d - i_{a.\nextw}$ is upper bounded by the number of deletions in block $a$. Likewise, $i_{a.\nextw} - d - i_a$ is upper bounded by the number of insertions between $i_a$ and $i_{a.\nextw}$ in $\cB$. Thus, adding the absolute values, we can bound this entire sum by $\ED(A, B)$. 
  
  Finally, $d|\{a : \mu(a) = [n-d+1, n]\}| \le d + \ED(A, B)$ as $d|\{a : \mu(a) = [n - d+1, n]\}|$ characters map correspond to at most $d$ characters in the optimal proticol $\mu'$.

  In total, we have that $\ED(\mu) \le 16\ED(A, B) + 12\frac{n}{d}\gamma + 4d$, as desired.
\end{proof}

\subsubsection*{Reduction to low-skew mappings}

For $D \ge 1$, we say that a monotone map $\mu : \cA \to \cB \cup \{\perp\}$ has \emph{skew} at most $D$ if for all $a_1, a_2 \in \mu$ we have that
\begin{align}
  \frac{1}{D} |s(a_1) - s(a_2)| \le |s(\mu(a_1)) - s(\mu(a_2))| \le D |s(a_1) - s(a_2)| \label{eq:skew}
\end{align}
We let $D(\mu)$ be the minimum $D$ such that $\mu$ has skew at most $D$. See Figure~\ref{fig:skew} for a depiction of large and small skew.

\begin{figure*}[t!]
  \centering
  \begin{subfigure}{0.9\textwidth}
    \centering
    \includegraphics[width=6in]{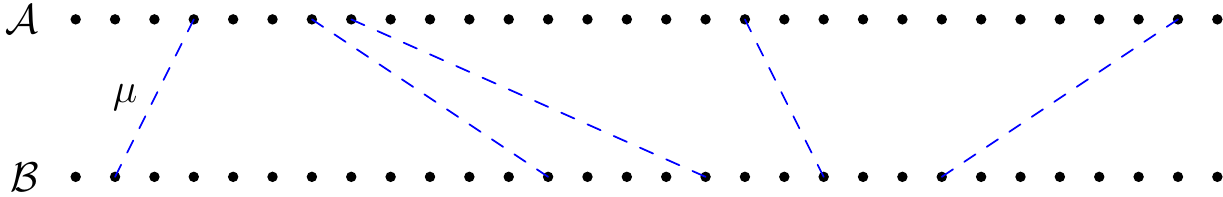}
    \caption{Large skew.}
    \label{fig:big-skew}
  \end{subfigure}\\
  \vspace{.5in}
  \centering
  \begin{subfigure}{0.9\textwidth}
    \centering
    \includegraphics[width=6in]{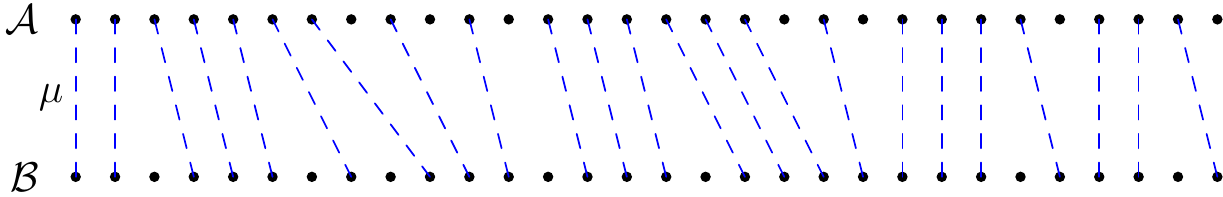}
    \caption{Small skew.}
    \label{fig:small-skew}
  \end{subfigure}\\
  \caption{(a) A matching $\mu$ with $D(\mu) > 2$. Because of the large skew, many edges cannot be matched by $\mu$. Therefore, deleting the edges contributing to large skew cannot change the edit distance by much. (b) A matching $\mu$ with $D(\mu) < 2$.}
  \label{fig:skew}
\end{figure*}

Although a minimal choice of $\mu$ with respect to $\ED(\mu)$ may have arbitrarily large skew, we show that there exists $\mu' \subset \mu$ whose skew is at most two and $\ED(\mu')$ is within a constant factor of $\ED(\mu)$.

\begin{claim}\label{claim:low-skew}
  For all monotone mappings $\mu : \cA \to B\cup \{\perp\}$, there exists $\mu' \subset \mu$ such that $\ED(\mu') \le 9\ED(\mu)$ and $D(\mu') \le 2$.
\end{claim}

\begin{proof}
  Assume that $D = 2$ for the remainder of this proof. Let $S \subset \cA \times \cA$ be the set of all pairs $(a_1, a_2)$ which violate (\ref{eq:skew}) and $s(a_1) \le s(a_2)$. We put a partial ordering $\preceq$ on $S$ such that $(a_1, a_2) \preceq (a_3, a_4)$ if $s(a_3) \le s(a_1) \le s(a_2) \le s(a_4)$.
 Let $S' \subset S$ be the set of pairs which are maximal with respect to $\preceq$.

  We also say that two pairs $(a_1, a_2)$ and $(a_3, a_4)$ are \emph{disjoint} if $s(a_2) \le s(a_3)$ or $s(a_4) \le s(a_1)$.

  Consider the following procedure to build a set $T$ of pairwise disjoint elements of $S$. First, take any maximal element $(a, a')$ of $S$ and insert it into $T$. Then, delete from $S$ any pairs which are not disjoint from $(a, a')$. Continue these two steps until $S$ is empty. 

  Label the elements of $T$,  $(a_1, a_2), \hdots, (a_{2k-1}, a_{2k})$ such that $s(a_1) \le s(a_2) \le \cdots \le s(a_{2k-1}) \le s(a_{2k})$. For all windows $a$ such that $s(a_{2i-1}) \le s(a) < s(a_{2i})$ for some $i$, set $\mu'(a) = \perp$, Otherwise, set $\mu'(a) = \mu(a)$. Note that $\mu'$ has skew at most $2$ because for every pair $(a, a') \in S$, $(a, a')$ must be not disjoint from $(a_{2i-1}, a_{2i})$ for some $i$. Because $(a_{2i-1}, a_{2i})$ was chosen maximally, $(a_{2i-1}, a_{2i})$ cannot be contained in $(a, a')$.  In particular, one of $a$ or $a'$ is in the interval from $(a_{2i-1}, a_{2i})$. (If it were the case that $a' = a_{2i}$, then $a$ must be between $a_{2i-1}$ and $a_{2i}.\prevw$ and so is deleted.) Thus, at least one element of the pair maps to $\perp$ in $\mu'$. This means that $\mu'$ has skew at most $2$.

  Now we show that $\ED(\mu') = O(\ED(\mu))$. Let $m$ be the total number of windows deleted in the previous step. First, we show that $\ED(\mu') \le \ED(\mu) + 4md$, and then we show that $md \le 2\ED(\mu)$.

  To show the first inequality, it suffices to show that for any $\nu, \nu' : \cA \to \cB$ monotone, $\ED(\nu') - \ED(\nu') \le 2(|\nu| - |\nu'|)d$, where $|\nu|$ is the number of windows mapped to something other than $\perp$. In particular, by a simple inductive argument, it suffices to consider the case $|\nu| - |\nu'| = 1$. That is, there is exactly one $a$ such that $\nu(a) \neq \perp$ but $\nu'(a) = \perp$. If $a$ is either the first of the last window, then
  \[
    \ED(\nu') - \ED(\nu) \le 2\ED(a, \perp) - 2\ED(\nu(a), \perp) \le 2d.
  \]
  Otherwise, if $a.\prevw$ and $a.\nextw$ both exist (with respect to $\nu$), then
  \begin{align*}
    \ED(\nu') - \ED(\nu) &= 2\ED(a, \perp) - 2\ED(\nu(a), \perp) + 2|s(\mu(a.\prevw)) + d - s(\mu(a.\nextw))|\\&\ \ \ \ - 2|s(\mu(a.\prevw)) + d - s(\mu(a))| - 2|s(\mu(a)) + d - s(\mu(a.\nextw))|\\
                         &\le 2d + 2|s(\mu(a.\prevw)) + d - s(\mu(a.\nextw))\\ &\ \ \ \  - 2[s(\mu(a.\prevw)) + d - s(\mu(a))] - 2[s(\mu(a)) + d - s(\mu(a.\nextw))]|\\
                         &= 4d.
  \end{align*}
  This proves the base case, and thus $\ED(\mu') \le \ED(\mu) + 4md$, as desired.

  Now we see to show that $md \le \ED(\mu)$.  Let $m_i$ be the number of windows deleted between $a_{2i-1}$ and $a_{2i}$ (inclusive). Thus, $s(a_{2i}) - s(a_{2i-1}) \ge m_id$. Because $(a_{2i-1}, a_{2i})$ violates (\ref{eq:skew}), we have that either
  \[
    s(\mu(a_{2i})) - s(\mu(a_{2i-1})) \ge 2(s(a_{2i}) - s(a_{2i-1})) \ge 2m_id.
  \]
  or
  \[
    s(\mu(a_{2i})) - s(\mu(a_{2i-1})) \le \frac{1}{2}(s(a_{2i}) - s(a_{2i-1}))
  \]

  In the first case,
  \begin{align*}
    &2\sum_{\substack{a \in \cA, \mu(a) \neq \perp\\s(a_{2i-1}) \le s(a) < s(a_{2i})}} |s(\mu(a)) + d - s(\mu(a.\nextw))|\\
    &\ge 2\sum_{\substack{a \in \cA, \mu(a) \neq \perp\\s(a_{2i-1}) \le s(a) < s(a_{2i})}} [s(\mu(a.\nextw)) - d - s(\mu(a))]\\
    &\ge 2[s(\mu(a_{2i})) - s(\mu(a_{2i-1})) - m_id]\\
    &\ge 2m_id.
  \end{align*}

  In the second case, let $\ell_i$ be the number of windows between $a_{2i-1}$ and $a_{2i}$ for which $\mu(a) = \bot$. Note that $(\ell_i + \mu_i)d = s(a_{2i}) - s(a_{2i-1}).$

  Observe that
  \begin{align*}
    &2\sum_{\substack{a \in \cA, \mu(a) \neq \perp\\s(a_{2i-1}) \le s(a) < s(a_{2i})}} \ED(a, \mu(a)) + 2\sum_{\substack{a \in \cA, \mu(a) \neq \perp\\s(a_{2i-1}) \le s(a) < s(a_{2i})}} |s(\mu(a)) + d - s(\mu(a.\nextw))|\\
    &\ge 2\ell_i d + 2\sum_{\substack{a \in \cA, \mu(a) \neq \perp\\s(a_{2i-1}) \le s(a) < s(a_{2i})}} [s(\mu(a)) + d - s(\mu(a.\nextw))]\\
    &\ge 2\ell_i d + m_i d - s(\mu(a_{2i})) + s(\mu(a_{2i-1}))\\
    &\ge \frac{1}{2}(2\ell_i + m_i)d\ge \frac{1}{2}m_i d.
  \end{align*}
  
  Note in particular this means that
  \begin{align*}
    \ED(\mu) &\ge \sum_{i = 1}^k \frac{1}{2}m_i d=\frac{1}{2}md,
  \end{align*}
  as desired. Thus, $\ED(\mu') \le 5\ED(\mu).$
\end{proof}

\begin{remark}
  Essentially the same proof works for any $D > 1$, replacing the constant factor of $9$ with a suitable function of $D$.
\end{remark}

\section{Reduction to the Query Problem}\label{sec:alg}

In Section~\ref{sec:windows}, we have reduced approximating $\ED(A, B)$, to finding a low-skew matching $\mu$ between the $\cA$ windows and $\cB$ windows which (approximately) minimizes $\ED(\mu)$. Note that we incur an additive error in our approximation due to the discretization of $A$ and $B$ into windows $\cA$ and $\cB$. 

In Section~\ref{sec:41}, we demonstrate how to compute $\ED(\mu)$ given estimates of the edit distance between windows of $\cA$ and windows of $\cB$ (Algorithm~\ref{alg:edfromest}).  In Section~\ref{sec:42}, we reduce (Algorithm~\ref{alg:estimate}) computing these pairwise window edit-distances to an algorithm \textsc{Query} which is stated an analyzed in Section~\ref{sec:query}. In Section~\ref{sec:43}, we state and analyze \textsf{MAIN} (Algorithm~\ref{alg:alg2query}), which takes the original strings as input, partitions them into windows and calls the other algorithms. See Figure~\ref{fig:algflow} to see how these algorithms interconnect.

\subsection{Reduction to Estimating Window Distances}\label{sec:41}

In order to optimize $\ED(\mu)$, we compute for every pair $(a,b) \in \cA \times \cB$ an estimate $\cE(a, b) \geq \ED(a, b)$ (the quality of the approximation will be discussed soon). For a given monotone mapping $\mu : \cA \to \cB \cup \{\perp\}$, we define the edit distance of this mapping with respect to this estimate $\cE$ to be 

\[
  \ED_{\cE}(\mu) := 2\sum_{a\in \cA} \cE(a,\mu(a))
  + 2\sum_{\substack{a \in \cA, \mu(a)\neq \perp\\a\text{ is not last}}} |s(\mu(a))+d-s(\mu(a.\nextw))|,\]

where $\cE(a, \perp) = d$ by definition. Further define
\[
  \ED(\cE) := \min_{\substack{\mu : \cA \to \cB \cup \{\perp\}\\\mu\text{ monotone}}}\ED_{\cE}(\mu).
\]

Note that the space needed to store $\cE$ is $|\cA| |\cB| = \frac{n}{d}\cdot \frac{n}{\gamma}$. In the regime we are working $(\Delta \ge n^{1-\epsilon} \gg n^{1-\delta})$, we will have that $d = \sqrt{n}$ and $\gamma > n^{1/2 - \delta}$, so the total storage is $O(n^{1+\delta})$.

Given such an estimate $\cE$, we can efficiently compute an optimal monotone matching for this objective. Our notion of window edit distance is slightly different from that of~\cite{BEGHS18-edit-quantum} (see Lemma~4.1 of their paper), but both are consistent within a constant factor (and sublinear additive term). 
For completeness, we include the pseudocode below which runs in $O(\frac{n^2}{\gamma}^2)$ time.

\begin{algorithm}
  \caption{Computing Edit Distance from Estimates: $\textsc{EDfromEstimates}(\cA, \cB, \cE, n, d, \gamma)$}

\begin{algorithmic}
  \STATE $M \leftarrow 0^{(|\cA|+1) \times (|\cB|+1) \times (s+1)}$; \COMMENT{three-dimensional state array}
  \STATE $s \leftarrow d/\gamma$;
  \STATE \textbf{for} $i = 1$ to $|\cA|$ \textbf{do}\\
  \STATE \ \ \ \textbf{for} $j = 1$ to $|\cB|$ \textbf{do}\\
  \STATE \ \ \ \ \ \ \textbf{for} $k = 0$ to $s$ \textbf{do}\\
  \STATE \ \ \ \ \ \ \ \ \ $M(i, j, k) \leftarrow M(i - 1, j, k) + 2\cE(\cA[i], \cB[j])$;
  \STATE \ \ \ \ \ \ \ \ \ \textbf{if} $k = 0$ \textbf{then}\\
  \STATE \ \ \ \ \ \ \ \ \ \ \ \ \textbf{for} $\ell = 1$ to $s$ \textbf{do}\\
  \STATE \ \ \ \ \ \ \ \ \ \ \ \ \ \ \ $M(i, j, k) \leftarrow \min(M(i, j, k), M(i-1, j-1, \ell) + (s - \ell)\gamma)$;\\
  \STATE \ \ \ \ \ \ \ \ \ \textbf{if} $k \in \{1, 2, \hdots, s-1, s\}$ \textbf{then}\\
  \STATE \ \ \ \ \ \ \ \ \ \ \ \ $M(i, j, k) \leftarrow \min(M(i, j, k), M(i, j-1, k-1))$;\\
  \STATE \ \ \ \ \ \ \ \ \ \textbf{if} $k = s$ \textbf{then}\\
  \STATE \ \ \ \ \ \ \ \ \ \ \ \ $M(i, j, k) \leftarrow M(i, j-1, k) + 2\gamma$.  
  \RETURN $\min\{M(|\cA|, |\cB|, 0), \hdots, M(|\cA|, |\cB|, s)\}$;
  \end{algorithmic}
  \label{alg:edfromest}
\end{algorithm}

A key property of $\ED(\cE)$ is that it is monotonic in $\cE$. That is, if for all $a \in \cA$ and $b \in \cB$, we have that $\ED(a, b) \le \cE(a, b) \le \alpha \ED(a, b)$, then $\ED(\mu) \le \ED_{\cE}(\mu) \le \alpha \ED(\mu)$ for all $\mu$. In particular, this implies that $\ED(\cE)$ is within a constant factor of $\ED(A, B)$. Sadly, we are not aware of an algorithm which can guarantee that $\cE(a, b) \le \alpha \ED(a, b)$ for all $a$ and $b$. In the next subsection, we show a more subtle guarantee which suffices, when the edit distance is large.

\subsection{Obtaining an estimate: Algorithmic reduction to query model}\label{sec:42}

We claim that Algorithm~\ref{alg:estimate} allows us to reduce to working with estimates.

\begin{algorithm}[H]
  \caption{Obtain estimate of window distances: $\textsc{AggregateEstimates}(\cA, \cB, \cL, \Delta, L)$}
  \begin{algorithmic}
    \STATE \textbf{for} $i = 1$ to $n/d$ \textbf{do} \COMMENT{Initialize window distances}\\
    \STATE \ \ \ \textbf{for} $j = 1$ to $n/\gamma$ \textbf{do}\\
    \STATE \ \ \ \ \ \ $\cE(\cA[i], \cB[j]) \leftarrow d$;\\
    \STATE \textbf{for} $\Delta_{query} \in \{2^0\gamma, 2^1\gamma, 2^2\gamma, \hdots, d\}$ \textbf{do}\\
    \STATE \ \ \ $\hat{E}\leftarrow \textsc{Query}(\cA,\cB, \cL, 0, \cA\cup \cB \Delta_{query}, L)$; \COMMENT{Approximate pair-wise distances}\\
    \STATE \ \ \ \textbf{for} $i = 1$ to $n/d$ \textbf{do}\\
    \STATE \ \ \ \ \ \ \textbf{for} $j = 1$ to $n/\gamma$ \textbf{do}\\
    \STATE \ \ \ \ \ \ \ \ \ \textbf{if} $(\cA[i], \cB[j]) \in \hat{E}$ \textbf{then}\\
    \STATE \ \ \ \ \ \ \ \ \ \ \ \ $\cE(\cA[i], \cB[j]) \leftarrow \min(\cE(\cA[i], \cB[j]), \beta_L \Delta_{query})$;\\
    \RETURN $\cE$;
  \end{algorithmic}
  \label{alg:estimate}
\end{algorithm}

The $\textsc{AggregateEstimates}$ procedure works by splitting the task into easier questions. Given a threshold $\Delta_{query}$, which pairs in $\cA \times \cB$ have edit distance at most $O(\Delta_{query})$? To determine this, \textsc{ObtainEstimate} makes calls to $\textsc{Query}(\cA, \cB, \Delta_{query}, L)$ (see Algorithm~\ref{alg:query-recursion}). We assume that $\textsc{Query}$ has the following guarantees.

\begin{lemma}[Main query algorithm]\label{lem:main-query}
  For all $L, i_{\max} \ge 1$, let $\eps= \frac{1}{200^{L + i_{\max} + 1}}$ and $t_{\min} = (1000/\eps^{10})^{4^{L+1}/\eps^2}$.
Assume $t \ge t_{\min}$, $\Delta_{query} \ge [d^{1-\eps}, d]$.

  $\textsc{Query}(\cA, \cB, \Delta_{query}, L)$ (Algorithm~\ref{alg:query-recursion}) makes at most $t^{1+3/i_{\max}}$ queries to $\textsc{Main}(a, b, \Delta_{query}, L-1)$ and outside of those queries runs in $t^{2+3/i_{\max}}$ time.

  Assume that for  every $(a, b)$ for which $\textsc{Main}(a, b, \Delta_{query}, L-1)$ is called, the Algorithm $\textsc{Main}$ correctly returns
  \begin{gather}\label{eq:main-succeed}
    \widehat{\ED}(a, b) \in [\ED(a, b), \alpha_{L-1}(\ED(a, b) + \Delta_{query})].
  \end{gather}
  For all monotone $\mu \subset E$ such that $D(\mu) \le 2$ with probability at least $1 - e^{-t^{\eps}}$, over the remaining randomness, $\textsc{Query}$ returns an edge set $\hat{E}$ such that
  \begin{itemize}
  \item For all $(a, b) \in \hat{E},$ $\ED(a, b) \le \beta_L\Delta_{query}$.
  \item 
    $|\{(a, b) \in \mu \setminus \hat{E} : \ED(a, b) \le \Delta_{query}\}| \le t^{1-\eps}.$
  \end{itemize}
\end{lemma}

The algorithm and analysis of $\textsc{Query}$, including proving Lemma~\ref{lem:main-query}, are described in Section~\ref{sec:query}.

Using Lemma~\ref{lem:main-query}, we can prove the following guarantee for Algorithm $\textsc{AggregateEstimates}$.

\begin{lemma}\label{lem:main-estimate}
  For all $L, i_{\max} \ge 1$, let $\epsilon' = \frac{1}{200^{L + i_{\max} + 2}}$ and $n_{\min} = (1000/(\eps')^{10})^{4^{L+2}/(\eps')^2}$.

  For all $n \ge n_{\min}$ and $\Delta \in [n^{1-\eps'}, n]$, $\textsc{AggregateEstimates}(\cA, \cB, \Delta, L)$ makes at most $n^{1/2 + 2/i_{\max}}$ queries to $\textsc{Main}(a, b, \cdot, L-1)$ and otherwise runs in time $n^{1+2/i_{\max}}.$ Assuming that all of these calls to $\textsc{Main}$ succeed (see Eq.~\eqref{eq:main-succeed}), then $\textsc{AggregateEstimates}$ with probability $1 - e^{-t^{\eps'}}$returns an estimate $\cE$ such that  
  \[
    \ED(A, B) \le \ED(\cE) \le \alpha_L(\ED(A, B) + \Delta).
  \]
\end{lemma}

\begin{proof}
  Apply Lemma~\ref{lem:main-query}. Note that $d = n^{1/2}$, $\gamma \in [n^{1/2-2\eps'}, n^{1/2}]$, and $t \in [n^{1/2}, n^{1/2+\eps'}].$ Since $\Delta_{query} \in [\gamma, d] = [d^{1-\eps'}, d]$, the procedure $\textsc{Query}$ makes $t^{1 + 3/i_{\max}}$ queries and runs in $t^{2+3/i_{\max}}$ time (outside of queries). Since there are at most $\log_2(d/\gamma) + 1$ queries, the total number of queries is at most $(\log_2(d/\gamma) + 1)t^{1+3/i_{\max}} \le n^{1/2 + 2/i_{\max}}$ and the total other running time is at most $(\log_2(d/\gamma) + 1)t^{2 + 3/i_{\max}} \le n^{1 + 2/i_{\max}}$.

  This shows that $\textsc{AggregateEstimates}$ is efficient. Now we show correctness, assuming all recursive calls to $\textsc{Main}$ succeed and that $\textsc{Query}$ succeeds as well. Let $\hat{E}_{\Delta_{query}}$ be the edge set returned by $\textsc{Query}(\cA, \cB, \Delta_{query}, L)$. Let $\mu_{\Delta_{query}}$ be the set of $a \in \mu$ for which $\ED(a, \mu(a)) \le \Delta_{query}.$

  Then, by the first guarantee of Lemma~\ref{lem:main-query}, 
  for all $(a, b) \in \cA \times \cB$,
  \[
    \cE(a, b) \min_{(a, b) \in \hat{E}_{\Delta_{query}}} \beta_L\Delta_{query} \ge \ED(a, b).
  \]
  This immediately guarantees that $\ED(\cE) \ge \ED(A, B)$. To show the upper bound on $\ED(\cE)$, by Lemma~\ref{prop:ED-mu-upper-bound} and Claim~\ref{claim:low-skew}, we have that there exists $\mu \subset E$ with $D(\mu)\le 2$ and
  \[
    \ED(\mu) \le 144\ED(A, B) + 108\frac{n}{d}\cdot \gamma + 36d \le 144\ED(A, B) + 144\Delta.
  \]
  Define $\mu_{\gamma/2} = \emptyset$ and $\hat{E}_{\Delta_{query}} = \emptyset$ for notational convenience. We have that
  \begin{align}
    \ED(\cE) &\le \ED_{\cE}(\mu)\nonumber\\
             &\le \ED(\mu) + 2\sum_{a \in \mu} \cE(a, \mu(a))\nonumber\\
             &= \ED(\mu) + 2\sum_{a \in \mu} \min_{(a, \mu(a)) \in \hat{E}_{\Delta_{query}}} \beta_L \Delta_{query},\label{eq:aac}
    \end{align}
    by definition of $\cE$. Observe
    \begin{align}
    \sum_{a \in \mu} \min_{(a, \mu(a)) \in \hat{E}_{\Delta_{query}}} \beta_L \Delta_{query} &\le \sum_{\Delta_{query} \in [\gamma, 2\gamma, \hdots, d]} \sum_{(a, \mu(a)) \in \hat{E}_{\Delta_{query}} \setminus \hat{E}_{\Delta_{query}/2}} \beta_L \Delta_{query}\nonumber\\
             &= \sum_{\Delta_{query} \in [\gamma, 2\gamma, \hdots, d]} |\mu \cap (\hat{E}_{\Delta_{query}} \setminus \hat{E}_{\Delta_{query}/2})|\beta_L\Delta_{query}\nonumber\\
             &\le \sum_{\Delta_{query} \in [\gamma, 2\gamma, \hdots, d]} (|\mu _{\Delta_{query}} \setminus \mu_{\Delta_{query}/2}| + t^{1-\eps})\beta_L\Delta_{query}\nonumber\\
             &\le  \sum_{(a, \mu(a)) \in \mu} (\beta_L\ED(a, \mu(a)) + \gamma) + (2d) \cdot t^{1-\eps} \beta_L\label{eq:aacc}
    \end{align}
             
    Combingin Eq.~(\ref{eq:aac}) and Eq.~(\ref{eq:aacc}),
    \begin{align*}
             \ED(\cE) &\le \ED(\mu) + 2 \left[\sum_{(a, \mu(a)) \in \mu} (\beta_L\ED(a, \mu(a)) + \gamma) + (2d) \cdot t^{1-\eps} \beta_L\right]\\
             &\le (4\beta_L + 1)\ED(\mu) + 4d(t^{1-\eps} + \gamma)\beta_L\\
             &\le (4\beta_L + 1)\ED(\mu) + \beta_L n^{1-\eps'}\\
             &\le (576\beta_L + 144)\ED(A, B) + (577\beta_L + 144) \Delta\\
             &= \alpha_L (\ED(A, B) + \Delta).
  \end{align*}
  Each of these $\hat{E}_{\Delta_{query}}$ satisfies the needed guarantee with probability $1 - e^{-t^{\eps}}$, so by a union bound they all succeed with probability at least $1 - e^{-t^{\eps'}}$. 
\end{proof}

\subsection{Main algorithm}\label{sec:43}

We can now succinctly describe the main algorithm.

\begin{algorithm}[H]
\caption{Reduction to query algorithm: $\textsc{Main}(A, B, \Delta, L)$}
\begin{algorithmic}
  \INPUT Strings $A,B \in \Sigma^n$, threshold $\Delta$, recursion level $L$\\
  \STATE \textbf{if} $L = 0$ \textbf{then}\COMMENT{Base case of the recursion}\\
  \STATE \ \ \ \textbf{run} $O(n^2)$ algorithm for $\ED(A, B)$
  \STATE \ \ \ \textbf{return} $\ED(A, B)$.
  \STATE $d \leftarrow \sqrt{n}$; \\
  \STATE $t \leftarrow n/d$;
  \STATE $\gamma \leftarrow \max\{1,\Delta d / n\}$;\\
  \STATE $\cA \leftarrow \{A[1,d], A[d,2d],\dots A[n-d+1,n]\}$;\\
  \STATE $\cB \leftarrow \{B[1,d], B[1+\gamma,d+\gamma],\dots B[n-d+1,n]\}$;\\
  \STATE $\cL \leftarrow \emptyset$;\\
  \STATE \textbf{for} $\ell = 1, t^{\eps}, ... t$ \textbf{do}\\
  \STATE \ \ \ $\cL \leftarrow \cL\cup \{\cA[1, \ell], \cA[\ell+1, 2\ell], \hdots\}$;\\
  \STATE \ \ \ $\cL \leftarrow \cL\cup \{\cB[1, \ell], \cB[\ell+1, 2\ell], \hdots\}$;\\
  \STATE $\cE \leftarrow \textsc{AggregateEstimates}(\cA, \cB, \cL, \Delta, L)$;\\
  \RETURN $\textsc{EDfromEstimates}(\cA, \cB, \cE, n, d, \gamma)$;
  \end{algorithmic}
  \label{alg:alg2query}
\end{algorithm}

As a Corollary of Lemma~\ref{lem:main-estimate}, we have that $\textsc{MAIN}$ runs efficiently, without accounting for the recursion.

\begin{corollary}\label{cor:main-recursive}
  If $L = 0$, then $\textsc{Main}(A, B, \Delta, L)$ exactly computes $\ED(A, B)$ in $O(n^2)$ time.

  For all $L, i_{\max} \ge 1$, there exists $\eps'(L) = \frac{1}{200^{L + i_{\max} + 2}}$ and $n_{\min} = (1000/\eps'^{10})^{4^{L+2}/(\eps')^2}$ with the following property.

  For all $n \ge n_{\min}$ and $\Delta \in [n^{1-\eps'}, n]$, $\textsc{MAIN}(\cA, \cB, \Delta, L)$ makes at most $n^{1/2 + 2/i_{\max}}$ queries to $\textsc{Main}(a, b, -, L-1)$ and otherwise runs in time $n^{1+2/i_{\max}}.$ With probability at least $1 - e^{-t^{\eps'/2}}$ all these recursive calls succeed (and run efficiently) and the algorithm returns an estimate $\widehat{\ED}(A, B)$ such that
  \[
    \ED(A, B) \le \widehat{\ED}(A, B) \le \alpha_L(\ED(A, B) + \Delta).
  \]
\end{corollary}

\begin{proof}
  We induct on $L$. The base case the $L=0$ is obvious. Assume this holds for the case of $L-1$. Then, note that $\Delta_{query} \ge d^{1-2\eps(L)} \ge d^{1-\eps(L-1)}$ for all calls to $\textsc{MAIN}(-, L-1)$ satisfy the conditions of the inductive hypothesis. Thus, by Lemma~\ref{lem:main-estimate}, $\ED(\cE)$ has the prescribed properties and run-time with probability $1 - e^{-t^{\eps'(L)}}$, assuming all the recursive calls succeed. Since there are at most $t^{1+3/i_{\max}}$ recursive calls, each succeeding with probability at least $1 - e^{-t^{\eps'(L-1)/2}}$, the total success probability is at least
  \[
    1 - e^{-t^{\eps'(L)}} - t^{1 + 3/i_{\max}}e^{-t^{\eps'(L-1)/2}} \ge 1 - e^{-t^{\eps'(L)/2}}.
  \]

  Finally, by Algorithm~\ref{alg:edfromest}, $\ED(\cE)$ can be computed in $O(\frac{n^2}{\gamma^2}) \ll n^{1+2/i_{\max}}$ time, as desired.
\end{proof}

Now that we understand the recursive structure of $\textsc{Main}$, we can prove our main result.

\begin{theorem}[correctness and efficiency of $\textsc{Main}$]
  For all $L_{\max} \ge 0$ and $i_{\max} \ge 1$, there exists $\epsilon' = 1/200^{L_{\max} + i_{\max} + 2} > 0$ such that for all $n \ge (1000/\eps^{10})^{4^{L+2}/\eps^2}$ and any $\Delta \in [n^{1-\eps}, n]$, with probability $1 - e^{-t^{\eps'/2}} \ge 2/3$,  $\textsc{MAIN}(A, B, \Delta, L)$ returns an estimate $\widehat{\ED}(A, B)$ such that
  \[
    \ED(A, B) \le \widehat{\ED}(A, B) \le \alpha_L(\ED(A, B) + \Delta).
  \]
  Further, with high probability, this algorithm runs in $n^{1 + 1/2^L + 5/i_{\max}}$ time.
\end{theorem}

\begin{proof}
  We induct on $L$. The base case of $L = 0$ is obvious. 

  For $L \ge 1$, consider the $\eps$ from Corollary~\ref{cor:main-recursive}. We then have by the inductive hypothesis that with high-probability that the algorithm succeeds and the run-time is at most
  \[
    n^{1 + 2/i_{\max}} + n^{1/2 + 2/i_{\max}} d^{1 + 1/2^{L-1} + 5/i_{\max}} \ll n^{1 + 1/2^l + 5/i_{\max}},
  \]
  as desired.
\end{proof}

Setting $L_{\max} = \lfloor\log_2(2/\delta)\rfloor$ and $i_{\max} = \lfloor 10/\delta\rfloor$  so that $\delta \ge \frac{1}{2^L} + \frac{5}{i_{\max}}$ proves Theorem~\ref{thm:main}.

Notice then that
\begin{align*}
  \epsilon' &\approx 2^{-O(1/\delta)}\\
  n_{\min} &\approx 2^{2^{O(1/\delta)}}\\
  \alpha_{L_{\max}} &\approx 2^{2^{O(1/\delta \cdot \log(1/\delta))}}.
\end{align*}

\section{Query Model and Analysis}\label{sec:query}

\subsection{Roadmap}\label{subsec:query-roadmap}
The goal of this section is to prove Lemma~\ref{lem:main-query}, which says that, with respect to any matching $\mu$, Algorithm \textsc{Query} returns a set of edges $E$ which satisfies:
$$\{\text{Most of $\mu$}\} \subseteq E \subseteq \{\text{pairs of windows of edit distance $O(\Delta)$}\}.$$ 

We first describe an implementation of \textsc{Query} in Section~\ref{subsec:query-alg}.

It is helpful to consider the {\em algorithm tree} of recursive calls made by \textsc{Query}.
At each  $i$-th level internal node (or call to \textsc{Query}) we keep track of a set of {\em live windows} $\cI^i \subset \cA\cup \cB$, and the goal is to find a smaller set $\cI^{i+1} \subset \cI^i$ of live windows for the next round. To do this, we pick a uniformly random live window $x \in \cI^i$, and recursively call $\textsc{MAIN}(x, \cdot, \Delta, L-1)$, to compute the (approximate) edit distance from $x$ to every other window in $\cI^i$. The set of windows which are $O(\Delta)$ distance from $x$ is called a \emph{clique}, which we denote by $C^{i+1}$. If the clique is extremely big (see Figure~\ref{fig:big-clique}), that is $|C^{i+1}|/|\cI^i| \ge t^{-1/i_{\max}}$, then we terminate the recursion, adding  $C^{i+1} \times C^{i+1}$ to the output $E$. 
 Otherwise, if $C^{i+1}$ is relatively small, we define $\cI^{i+1}$ to be the set of windows that are close (in terms of their respective locations on the strings) to the windows in $C^{i+1}$ (see Figure~\ref{fig:small-clique}). Here ``close'' is chosen so that $|\cI^{i+1}|/|\cI^i| \approx t^{-1/i_{\max}}$. 
 If we haven't terminated earlier due to a big clique, by level $i_{\max}$ the set of live windows is sufficiently small so that we can brute force query all the pairs to compute their (approximate) edit distance.

At each node, we don't just do the refinement of live windows once, but rather we pick many candidate $x$'s (approximately $t^{1/i_{\max}}$ many) and branch on all of these. We show that with high probability, most edges of the matching $\mu$ will survive to the next level of the recursion tree (see Lemma~\ref{lem:inductive}). By induction this will imply that nearly every edge of $\mu$ is ultimately included in $E$; (Section~\ref{sub:proof-main-query}).

In Section~\ref{subsec:query-analysis} we introduce some key definitions and propositions that we later use in the proof of our main inductive step.
One important concept is that of {\em good windows}; intuitively a window $y$ is good if it is live and its match $\mu(y)$ is also live. 
Our analysis only holds for {\em active nodes} of the algorithm tree, which are those where the fraction of good live windows is not too small.
The third key concept we introduce in this section is a partition of all the edges in the matching $\mu$ into a constant number of {\em colors}; roughly we say that two edges $(y,\mu(y))$ and $(z,\mu(z))$ are of the same color if the respective cliques (or edit distance balls) around $y$ and $z$ have identical statistics (to within a multiplicative factor of $t^{\epsilon}$). In particular, Proposition~\ref{prop:most-in-interval-in-big-color} asserts that in every interval around a typical good window $y$, the fraction of live windows of the same color is non-negligible.
Note that good windows, active nodes, and colors depend on $\mu$ and hence are only used in the analysis of the algorithm.\footnote{There are actually two notions of color, a \emph{window color} which only considers local statistics of the window itself and is independent of $\mu$ and \emph{matching color} which is a combination of the windows colors of a window and its match according to $\mu$. The latter notion crucially depends on $\mu$.}

We complete the proof of the main inductive step in Section~\ref{subsec:query-correct}. The key step, which we call the {\em Color Lemma} (Lemma~\ref{lem:color}) due to the heavy use of colors, asserts that most $i$-level-good windows $y$ have a not-too-small probability of staying alive in level $i+1$. 
To prove the Color Lemma (Section~\ref{sub:proof-color}), 
we consider an interval $J$ around $y$, of length chosen with respect to the clique around $y$.
We then consider the $z \in J$, such that  $(y,\mu(y))$ and $(z,\mu(z))$ are of the same color, take the cliques around each such $z$, and finally let $W$ denote the union of cliques. We argue that (i) $W$ is large (Claim~\ref{claim:size-W}), and (ii) whenever the algorithm uses a clique around some $w \in W$, the good window $y$ remains live for the next level (Claim~\ref{claim:density-W}). To analyze $W$, it is useful that the clique around each $w \in W$ is similar to the clique around $z$ (since $w$ and $z$ are close in edit distance), and the clique around $z$ is similar to the clique around $y$ (since they have the same color).

Section~\ref{subsec:runtime} concludes with a short proof that the \textsc{Query} is efficient both in run-time and in the number of recursive calls it makes to \textsc{Main}.

\subsection{Query algorithm}\label{subsec:query-alg}

In this section, we assume that $\eps = \frac{1}{200^{L_{\max} + i_{\max} + 1}}$, $t \ge t_{\min} := (1000/\eps^{10})^{4^{L+1}/\eps}$, and $\Delta \in [d^{1-\eps/2}, d]$ (representing the $\Delta_{query}$ from Algorithm~\ref{alg:alg2query}).  Let $E_{\Delta}$ be the bipartite graph on $(\cA; \cB)$ such that $(a,b) \in E_{\Delta}$ iff $\ED(a,b) \leq \Delta$.

For any edge set $E$ and monotone mapping $\mu : \cA \rightarrow \cB \cup\{\perp\}$, we abuse notation and say that $\mu \subseteq E$ if for every $a \in \cA$, either $(a,\mu(a)) \in E_{\Delta}$, or $\mu(a) = \perp$.

Recall that our main goal (Lemma~\ref{lem:main-query}) is to construct a set $\hat{E}$ such that $\hat{E} \subset E_{\beta_L \Delta}$ and $\hat{E} \setminus E_{\Delta}$ misses few ($t^{1-\eps}$) edges of the optimal matching $\mu$.

\subsubsection*{Algorithm Tree}
Our query algorithm is recursive (not to be confused with the recursion of Algorithm {\sc Main}). Each depth-$i$ call to {\sc Query} instantiates $t^{\eps_{i+1} + 1/i_{\max}}$ new depth-$(i+1)$ recursive calls. This implicitly defines a tree where we associate the edges with executions of {\sc Query} and the nodes with the state of the algorithm at each call (in particular, the set of live windows defined below).

\subsection*{Live windows}
At each iteration, the algorithm maintains a set of {\em live windows}, denoted $\cI^i \subset \cA \cup \cB.$ We denote $\cI^i_{\cA} := \cI^i \cap \cA$ and $\cI^i_{\cB} := \cI^i \cap \cB$. The intention is that for sizeable fraction of windows $y \in \cI^{i}_{\cA}$ we have $\mu(y) \in \cI^{i}_{\cB}$, and vice-versa. At the beginning of the algorithm, we have $\cI^0 = \cA \cup \cB$.

\subsubsection*{Cliques}

At each call to {\sc Query}, we pick a uniformly random live window $x \in \cI^i$ and query the edit distance between $x$ to all other windows in $\cI^i$.

\begin{definition}[Cliques]\label{def:clique}
  Given $x \in \cI^i$ and $\tau \in T := \{1, \hdots, \tau_{\max} := 1000/\eps^3\}$, we let $C^{i+1}(x, \tau) \subset \cI^i$ be a \emph{clique} defined by
  \[
    C^{i+1}(x, \tau) := \{y \in \cI^i : \widehat{ED}(x, y) \le \Delta c_L^{\tau}\},
  \]
  where $\widehat{ED}(x, y)$ is the value returned by $\textsc{MAIN}(x, y, \Delta, L-1)$, and $c_L = 100\alpha_{L-1} = 100 \cdot 2^{(20000/\eps^2)^L}$ (see Table~\ref{table:param}).
\end{definition}

\begin{remark}
By the triangle inequality, every pair of windows in $C^{i+1}(x, \tau)$ is $2\Delta c_L^{\tau+1}$-close. Thus, we can include an edge in $\hat{E}$ for every pair of windows in $C^{i+1}$. Hence, we call $C$ a {\em clique}.
\end{remark}

For the purposes of the query analysis, we assume that $\widehat{ED}(x, y)$ is always computed correctly by $\textsc{MAIN}$, that is
\[
  \ED(x, y) \le \widehat{ED}(x, y) \le \alpha_{L-1}(\ED(x, y) + \Delta).
\]

We also assume that $\widehat{ED}(x, y) = \widehat{ED}(y, x)$. We can ensure this by maintaining a hash table\footnote{For clarity of exposition, we do not include this step in the algorithm, but adding it is straightforward and can only improve the runtime.} of computed window edit distances and only running $\textsc{MAIN}$ on new pairs.

The algorithm will randomly sample a small number of $C^{i+1}(x, \tau)$ for uniformly random $x \in \cI^i$ and $\tau \in T$. During the analysis, we consider an ensemble of all such cliques (we can do this by pretending that $\widehat{ED}(x, y)$ was determined correctly for all $x$ and $y$, even though we only need to computer a small number in the algorithm).

\subsubsection*{Intervals}

An \emph{interval} $I \subset \cA \cup \cB$ is a continuous segment of windows. (Thus, either $I \subset \cA$ or $I \subset \cB$.) We let $\cL$ denote the set of intervals  of each length in $\{1, t^{\eps}, \hdots, ~t\}$ which are disjoint and cover $\cA \cup \cB$, where this $\eps$ is the same as in the condition $\Delta \ge n^{1-\eps}$.

We also maintain a set $\Lambda := \{1,7\}$ of \textbf{interval-multipliers}.\footnote{We also use $49$ as an interval multiplier in parts of the analysis, but not where $\Lambda$ is used.} For interval $I \in \cL^i$ and multiplier $\lambda \in \Lambda$, we let $\lambda I$ denote the interval of length $\lambda \cdot |I|$ centered at $I$. If the centered interval goes off the ends of $\cA$ or $\cB$, then we truncate appropriately. Thus, $|\lambda I| < \lambda |I|$ on occasion, but we always have that $|\lambda I| \ge |I|$.

\subsubsection*{Snapping and stability}

We also have a function $\snap(\ell)$ which rounds $\ell$ to the largest power of $t^{\eps}$ which is at least $\ell$. We say that $\snap(0) = 0$. Thus, if $\ell \in \{1, \hdots, t\}$, there are at most $1/\eps + 2 \le 2/\eps$ possible snapped values. In particular, this means that in a monotone sequence of length $\gg 2/\eps$, there must be (many) consecutive values which snap to the same value.

\begin{claim}[Stability]\label{claim:stability}
  Let $s_1, \hdots, s_{k} \in \{0, 1, \hdots, t\}$ be a monotone sequence ($s_{i+1} \ge s_{i}$ for all $i$ or $s_{i+1} \le s_i$ for all $i$). Then, for all but $10/\eps$ values of $i \in \{1, \hdots, k\}$,
  \begin{gather}\label{eq:stability}
    \snap(s_{i-2}) = \snap(s_{i-1}) = \cdots = \snap(s_{i+2}).
  \end{gather}
\end{claim}

\begin{proof}
  Let $B \subset \{1, \hdots, k-1\}$ be the set of indices for which $\snap(s_{i}) \neq \snap(s_{i+1})$. Since the $s_i$'s are monotonic and $\snap(\cdot)$ takes on at most $2/\eps$ values, $|B| \le \frac{2}{\eps}$. If for some $i$, Eq.~\eqref{eq:stability} is not satisfied, then either $i \in \{1, 2, k-1, k\}$ or $B \cap \{i-2, i-1, i, i+1\} \neq \emptyset$. Thus, the number of violating $i$ is at most $4|B| + 4 \le 10/\eps$. 
\end{proof}

\subsubsection*{Density}
For a clique $C^{i+1}$ we have a notion of density as we describe below.

The {\em global density} of a clique $C^{i+1} \subseteq \cI^i$ is given by\footnote{All snaps are to the closest power of $t^{\eps_0}$.}

\[\grho(C^{i+1}) \triangleq \frac{\snap(|C^{i+1}|)}{\snap(|\cI^i|)}.\]

Define the {\em local density} of $C^{i+1}$ on interval $I$ as

\[\rho^{\textrm{local}}_{I}(C^{i+1} ) \triangleq  \frac{\snap(|C^{i+1} \cap I|)}{\snap(|\cI^i \cap I|)}.\]

Now, the {\em relative density} of $C^{i+1}$ on $I$ is the ratio of the local and global densities:

\[\rho_{I}^{\textrm{relative}}(C^{i+1}) \triangleq \frac{\rho^{\textrm{local}}_{I}(C^{i+1} ))}{\rho^{\textrm{global}}(C^{i+1})} = \frac{\snap(|C^{i+1} \cap I|)}{\snap(|\cI^i \cap I|)} \cdot \frac{\snap(|\cI^i|)}{\snap(|C^{i+1}|)}. \]

Unlike the other notations of density, relative density can often be greater than $1$. We let $\rho := t^{1/i_{\max}}$ denote the density threshold for our algorithm. We say that $C^{i+1}$ is {\em dense} on $\lambda I$ if $\rho^{\textrm{relative}}_{\lambda I}(C^{i+1}) \geq \rho.$ 

With this concept of density, we define $\cI^{i+1}$, given a clique center $x \in \cI^i$ and distance theshold $\tau \in T$ to be
\[
  \cI^{i+1}(x, \tau) = \bigcup \{7 I : I \in \cL\text{ and }\rho^{\textrm{relative}}_{7 I}(C^{i+1}(x, \tau)) \ge \rho\}.
\]
In other words, $\cI^{i+1}$ is the union of all intervals $\lambda I$ which are $\rho$-dense. In the algorithm, $x \in \cI^i$ and $\tau \in T$ are sampled uniformly at random.

\subsubsection*{Big cliques}
Notice that if $C^{i+1}$ contains almost a constant fraction of the live windows, namely $$\grho(C^{i+1})  < 1/\rho,$$ then no interval can be dense. 
On the high level this means that (i) we can't make progress by recursing via ``seed-and-expand''; but also we don't need to because (ii) we're in the ``dense graph'' case (here density is relative to the remaining live windows).
In particular, we can simply add all the pairs from  $C^{i+1} \times C^{i+1}$ to $E_{\Delta}$ and backtrack in the algorithm tree. We call such cliques $C^{i+1}$ \emph{big cliques.}

\subsubsection*{The Query Algorithm}

The following is the Query Algorithm, broken up into four methods.

\begin{algorithm}[H]
  \caption{Multi-level query algorithm: $\textsc{Query}(\cA, \cB, \cL, i, \cI^i, \Delta, L)$}
  \begin{algorithmic}
    \INPUT Window sets $\cA, \cB$, $\cL$ intervals, $i$ current depth of {\sc Query} recursion, $\cI^i$ live windows, $\Delta$ target edit distance, $L$ current level of {\sc Main} recursion.
  \STATE $t \leftarrow |\cA| + |\cB|$.
  \STATE $E \leftarrow \emptyset$\\
  \STATE {\bf if} $i = i_{\max}$ {\bf then} \COMMENT{Stopping condition}
  \STATE \ \ \ {\bf for} $x, y \in \cI^i$ {\bf do}\\
  \STATE \ \ \ \ \ \ {\bf if} $\textsc{Main}(x, y, \Delta, L-1) \le \Delta c_{L}$\\
  \STATE \ \ \ \ \ \ \ \ \ {\bf then} $E \leftarrow E \cup \{(x, y)\}$
  \STATE {\bf else for} $j = 1, \hdots t^{\eps_{i+1}}\rho$ \textbf{do}\\
  \STATE \ \ \ $C^{i+1} \leftarrow \textsc{SampleClique}(\cI^{i}, \Delta, L)$\\
  \STATE \ \ \ $\cI^{i+1} \leftarrow \textsc{DenseInterval}(\cL, C^{i+1}, \cI^{i})$
  \STATE \ \ \ {\bf if} $\cI_{\cA} = C^{i+1}$ \COMMENT{Big clique regime}\\
  \STATE \ \ \ \ \ \ {\bf then} $E \leftarrow E \cup (C^{i+1} \times C^{i+1})$\\
  \STATE \ \ \  {\bf else}\\
  \STATE \ \ \ \ \ \ $E \leftarrow E \cup \textsc{Query}(\cA, \cB, \cL, i+1, \cI^{i+1}, \Delta, L)$\\
  \RETURN $E$;
  \end{algorithmic}
  \label{alg:query-recursion}
\end{algorithm}

\begin{algorithm}[H]
  \caption{Clique sampling algorithm: $\textsc{SampleClique}(\cI, \Delta, L)$}
  \begin{algorithmic}
    \INPUT Live windows $\cI$, target radius $\Delta$, \textsc{Main} recursion-level $L \ge 1$.
    \STATE $C \leftarrow \emptyset$
    \STATE Sample $x \sim \cI_{\cA}$ uniformly at random\\
    \STATE Sample $\tau \sim T$ uniformly at random\\
    \STATE \textbf{for} $y \in \cI$ \textbf{do}\\
    \STATE \ \ \ $\widehat{\ED}(x, y) \leftarrow \textsc{Main}(x, y, \Delta, L-1)$\\
    \STATE \ \ \ \textbf{if} $\widehat{\ED}(x, y) \le \Delta c_L^{\tau}$ \textbf{then}\\
    \STATE \ \ \ \ \ \ $C \leftarrow C \cup \{y\}$\\
    \RETURN $C$
  \end{algorithmic}
  \label{alg:5}
\end{algorithm}

\begin{algorithm}[H]
  \caption{Interval partition algorithm: $\textsc{DenseInterval}(\cL, C^{i+1}, \cI^i)$}
  \begin{algorithmic}
    \INPUT Intervals $\cL$, clique $C^{i+1}$, live windows $\cI^i \supset C^{i+1}$.
    \STATE $\grho \leftarrow \snap|C|/\snap|\cI^i|$
    \STATE \textbf{if} $\rho^{\textrm{global}} > \rho^{-1}$, \textbf{then return} $C$; \COMMENT{Big clique regime.}\\  
    \STATE $\cI^{i+1} \leftarrow \emptyset$\\
    \STATE \textbf{for} $I \in \cL$ \text{such that} $7I \cap \cI^i \neq \emptyset$ \textbf{do}\\
    \STATE \ \ \ \textbf{if} $\rrho_{7 I}(C) \ge \rho$ \textbf{then}\\
    \STATE \ \ \ \ \ \ $\cI^{i+1} \leftarrow \cI^{i+1} \cup (\cI^i \cap \lambda^i I)$
    \RETURN $\cI^{i+1}$.
  \end{algorithmic}
  \label{alg:6}
\end{algorithm}

\subsection{Tools for analysis of the query algorithm}\label{subsec:query-analysis}
In this section, we define a number of important concepts that are crucial for the analysis of our the Query Algorithm.

\begin{table}[h]
  \caption{Additional notation for analysis}
  \label{table:proof}
  \begin{center}
    \begin{tabular}{rcl}
      Term & Definition & Note\\\hline
      $\cI^i$ & The set of live windows & \\
      $\cI^i_{\cA}$ & $\cI^i \cap \cA$ &\\     
      $\cI^i_{\cB}$ & $\cI^i \cap \cB$ &\\
      $\cI^i_{dd}$ & $y \in \cI^i$ which satisfy the decreasing densities (Definition~\ref{def:decreasing-densities}) &\\
      $\cI^i_{good}$ & current set of ``good'' windows (see Definition~\ref{def:good}) & $\cI^i_{good} \subseteq \cI^i_{\cA}$\\
      $\cI^i_{big}$ & $x \in \cI^i_{good}$ for which $\grho(C^{i+1}(x, \tau_{\max}-1)) > \rho^{-1}$ &\\
      $\cI^i_{small}$ & $\cI^i_{good} \setminus \cI^i_{small}$ &\\
                      
      $y$ & typical window in $\cI^i_{good}$ & Want to show  $y \in \cI^{i+1}_{good}$\\
      $\phiw(y)$ & Window color of $y$ & See Definition~\ref{def:window-color}\\
      $\phim(y)$ & Matching color of $y$ & See Definition~\ref{def:color}\\
      $\cJ(w, \tau)$ & Set of maximal intervals $J$ such that $7 J$ is $\rho$-dense on $C^{i+1}(w, \tau^i)$ &\\
      $J(y, w, \tau)$ & Unique $J \in \cJ(w, \tau)$ with $y \in J$ ($\bot$ if none exists) &\\
      $W_{\lambda}(w, \tau, \tau')$ & $\bigcup_{z \in \phim(y) \cap \lambda J(y, y, \tau)} C^{i+1}(z, \tau')$ & \\
      $\widehat{W}(w, \tau)$ & $W_7(y, \tau, \tau-1) \cap W_7(\mu(y), \tau, \tau-1)$ &\\
    \end{tabular} 
  \end{center}
\end{table}

\subsubsection*{Notation for $\mu$}
Recall that $\mu$ is a large monotone matching with low skew in $\cA \times \cB$ that we would like to discover. Although $\mu$ is defined as a partial function from $\cA$ to $\cB$, we extend $\mu$ to be a partial involution of $\cA \cup \cB$. That is, $\mu(\mu(y)) = y$ for all $y$ for which $\mu(y) \neq \bot$. This abuse of notation keeps things symmetric.

For any set $S \subset \cA \cup \cB$, we define $\mu(S) = \{\mu(a) : a \in S, \mu(a) \neq \perp\}.$ In particular, $S \cap \mu(S)$, is the set of windows in $S$ which match to another window in $S$.

\subsubsection*{Properties of cliques}

The following proposition says the key properties of cliques that we shall use.

\begin{proposition}\label{prop:clique-properties}
  For all $x, y \in \cI^i$ and $\tau \in T$, we have the following properties.

  \begin{enumerate}
  \item $x \in C^{i+1}(x, \tau)$.
  \item If $\mu(x) \neq \bot$, then $\mu(x) \in C^{i+1}(x, \tau)$. 
  \item If $\tau \neq \tau_{\max}$, $C^{i+1}(x, \tau) \subset C^{i+1}(x, \tau+1)$.
  \item If $y \in C^{i+1}(x, \tau)$ then $x \in C^{i+1}(y, \tau).$
  \item If $y \in C^{i+1}(x, \tau)$ then $C^{i+1}(x, \tau) \subset C^{i+1}(y, \tau+1).$
  \end{enumerate}
\end{proposition}

\begin{proof}
  \begin{enumerate}
    \item Since $\ED(x, x) = 0$, we have $\widehat{\ED}(x, x) \le \alpha_{L-1}\Delta \le c_L^{\tau}\Delta$. Thus, $x \in C^{i+1}(x, \tau)$.
    \item Since $\ED(x, \mu(x)) \le \Delta$, we have $\widehat{\ED}(x, \mu(x)) \le 2\alpha_{L-1}\Delta \le c_L^{\tau}\Delta$.
    \item $\widehat{\ED}(x, y) \le c_L^{\tau}$ implies $\widehat{\ED}(x, y) \le c_L^{\tau+1}.$
    \item This follows precisely from $\widehat{\ED}(x, y) = \widehat{\ED}(y, x)$ because we memoize previous calls to $\textsc{Main}$.
    \item Since $y \in C^{i+1}(x, \tau)$, $\ED(x, y) \le \widehat{\ED}(x, y) \le \Delta c_L^{\tau}$. In particular, for all $z \in C^{i+1}(x, \tau)$,
      \[
        \widehat{\ED}(y, z) \le \alpha_{L-1} (\ED(y, z) + \Delta) \le \alpha_{L-1}(\ED(y, x) + \ED(x, z) + \Delta) \le 3\alpha_{L-1}c_L^{\tau}\Delta \le c_L^{\tau+1}\Delta.
      \]
    \end{enumerate}
\end{proof}

\subsubsection*{Cyclic indices of intervals}

Given two intervals $I, I' \in \cL$, we say that $I \equiv I'$ if $\snap |I| = \snap |I'|$ and $I, I \subset \cA$ or $I, I' \subset \cB$. Define the \emph{shift} of an interval $I$, $\shift(I)$, to be the number of intervals to the left of $I$ of the same length. In other words,
\[\shift(I) := \left\lfloor \frac{\text{index of first window of $I$}}{\snap |I|}\right\rfloor,\]

Given a multiplier $\lambda \in \Lambda$, we say that $\lambda I \equiv \lambda I'$ if $I$ and $I'$ have the same length, are both in $\cA$ or both in $\cB$, and $\shift(I) \equiv \shift(I') \mod \lambda$.

Note that for every $y \in \cL$ and $\lambda \in \Lambda$ and every shift $s_0 \mod \lambda$, there exists exactly one $\lambda I \in \lambda \cL$ (the set of intervals with multiplier $\lambda$) such that $\shift(\lambda I) \equiv s_0 \mod \lambda$ and $y \in \lambda I$. To do this, we need to include intervals in $\lambda \cL$ which are centered at an interval that is ``off the boundary.'' For example, the set of intervals of length $t^{\eps}$ with interval multiple $7$ is
\[
\{\mathcal A[1, t^{\eps}], \mathcal A[1, 2t^{\eps}], \mathcal A[1, 3t^{\eps}], \hdots, \mathcal A[1, 7t^{\eps}], \mathcal A[t^{\eps}+1, 8t^{\eps}], \hdots, A[t-7t^{\eps}+1, t], \hdots, \mathcal A[t-t^{\eps}+1,t]\}.
\]

\subsubsection*{Decreasing interval-densities}
In the previous section, we define the notion of density of a clique. In this section, we also need a notion of {\em live-density}. For any interval $\lambda I \in \lambda \cL$, we define its live-density to be the cardinality of the intersection divided by the length of the interval. That is,

\begin{align*}
  \eta_{\lambda I}(\cI^i) := \frac{|\cI^i \cap \lambda I|}{|\lambda I|}.
\end{align*}
 We would like to ensure for any $y$, if the interval around $y$ increases, then the density cannot increase substantially. This is formalized in the following definition.

\begin{definition}\label{def:decreasing-densities}
  We let $\cI^i_{dd} \subset \cI^i$ be the set of $y \in \cI^i$ such that for all $I, J \in \cL$ with $y \in \lambda I$ and $I \subsetneq J$,
  \begin{gather}\label{eq:decreasing-densities}
    \frac{\eta_{\lambda J}(\cI^i)}{\eta_{\lambda I}(\cI^i)} \le t^{3\eps_i}.
  \end{gather}
\end{definition}

Note that if $i = 0$, then $\cI^0_{dd} = \cI^0 = \cA \cup \cB$. In general enforcing this condition only requires discarding a small fraction of $\cI^i$.

\begin{claim}
  We have that
  \begin{gather}\label{eq:decreasing-densities-works}
    \frac{|\cI^i_{dd}|}{|\cI^i|} \ge 1 - t^{-2\eps_i}.
  \end{gather}
\end{claim}

\begin{proof}
  Define
  \[
    B_{I, J, \lambda} := \begin{cases}
      \cI^i \cap \lambda I & \text{Eq.~\eqref{eq:decreasing-densities} fails}\\
      \emptyset & \text{otherwise}.
    \end{cases}
  \]
  We then have that
  \[
    \cI^i \setminus \cI^i_{dd} = \bigcup_{I, J, \lambda} B_{I, J, \lambda}.
  \]
  Thus,
  \begin{align*}
    |\cI^i \setminus \cI^i_{dd}| &\le \sum_{I, J, \lambda} |B_{I, J, \lambda}|\\
                                 &< t^{-3\eps_i} \sum_{I, J, \lambda} \frac{|\lambda I|}{|\lambda J|} |\cI^i \cap \lambda J|\\
                                 &= t^{-3\eps_i} \sum_{J, \lambda} \frac{|\cI^i \cap \lambda J|}{|\lambda J|}\sum_{\lambda I \subset \lambda J} |\lambda I|\\
                                 &\le t^{-3\eps_i} \sum_{J, \lambda} \frac{|\cI^i \cap \lambda J|}{|\lambda J|}(\lambda \ell |\lambda J|)\\
                                 &\le 7(2/\eps) t^{-3\eps_i} \sum_{J, \lambda} |\cI^i \cap \lambda J|\\
                                 &\le 7(2/\eps)t^{-3\eps_i} \sum_{\lambda} \lambda (2/\eps) |\cI^i|\\
                                 &\le 56(2/\eps)^2 t^{-3\eps_i}|\cI^i|\\
                                 &\le t^{-2\eps_i}|\cI^i|.
  \end{align*}
  Note in the last step, we used $t \ge t_{\min}$. This implies Eq.~\eqref{eq:decreasing-densities-works}. 
\end{proof}

\subsubsection*{Good windows}

In our analysis, we maintain an additional set of windows $\cI^{i}_{good} \subseteq \cI^{i}$. Defined as follows

\begin{definition}[Good windows]\label{def:good}
  For a given node at level $i$ of the recursion tree, we define
  \begin{gather*}
    \cI^{i}_{good} := \cI^i_{dd} \cap \mu(\cI^i_{dd}).
  \end{gather*}
\end{definition}

$\cI^{i}_{good}$ roughly keeps track of which $y \in \cI^i$ also have $\mu(y) \in \cI^i$. Note that the algorithm does not ``know'' $\cI^{i}_{good}$ because it does not know $\mu$. We also require that $y$ and $\mu(y)$ satisfy the decreasing densities condition.

Note that our base case is
\begin{gather}\label{eq:good-base-case}
  \cI^{0}_{good} = \{y \in \cA \cup \cB : \mu(y) \neq \perp\}
\end{gather}

\subsubsection*{Active nodes}

Throughout the algorithm, we say that a node at level $i$ in the algorithm tree is \emph{active} if the following condition holds

\begin{definition}[Active node]\label{def:invariant}
  A node at level $i$ of the algorithm tree is \emph{active} if the following holds:
  \begin{align}
    \frac{|\cI^i_{good}|}{|\cI^i|} &\ge t^{-\eps_i}\label{eq:invariant-good-big}.
\end{align}
\end{definition}

Eq.~\eqref{eq:invariant-good-big} ensures that $\cI^i_{good}$ remains non-negligible compared to $\cI^i$.
Note that as long as the inactive nodes do not cause the algorithm to run for too long, they cannot hurt the quality of our solution.

\subsubsection*{Stable balls}

For a fixed $x \in \cI^i$, notice that
\[
  C^{i+1}(x, 1) \subset C^{i+1}(x, 2) \subset \cdots \subset C^{i+1}(x, \tau_{\max}).
\]
Thus, by Claim~\ref{claim:stability}, $\grho(C^{i+1}(x, \tau))$ is often constant in a subsequence around $\tau$. In fact, for any interval $\lambda I \in \lambda \cL$, we can similarly deduce that $\lrho_{\lambda I}(C^{i+1}(x, \tau))$ is often constant.

\begin{definition}[Stable balls]\label{def:tau-stable}
We say that $x \in \cI^i$ is {\em stable} at $\tau \in T$ (or $\tau$-stable) if for all $\lambda \in \Lambda$ and $\lambda I \subset \lambda \cL$ such that $x \in \lambda I$.
\[
  \lrho_{\lambda I}(C^{i+1}(x, \tau-2)) = \cdots = \lrho_{\lambda I}(C^{i+1}(x, \tau+2)).
\]
and
\[
  \grho(C^{i+1}(x, \tau-2)) = \cdots = \grho(C^{i+1}(x, \tau+2)).
\]
\end{definition}

By definition of relative density, any $\tau$-stable $x$ also satisfies
\[
  \rrho_{\lambda I}(C^{i+1}(x, \tau-2)) = \cdots = \rrho_{\lambda I}(C^{i+1}(x, \tau+2)).
\]

\begin{claim}[$x$ is often $\tau$-stable.]\label{claim:tau-stable}
  For all $x \in \cI^i$, $x$ is stable at $\tau$ for all but $200/\eps^2$ choices of $\tau \in T$.
\end{claim}

\begin{proof}
  For a given $x \in \cI^i$, notice that there are at most $2/\eps$ intervals $I \in \cL$ which contain $x$, and at most $7(2/\eps)$ intervals $7I \in 7\cL$ which contain $x$. In addition to the global density condition, there are $14/\eps + 2/\eps + 1 < 20/\eps$ separate density conditions to satisfy. By Claim~\ref{claim:stability}, any particular condition is violated for at most $10/\eps$ choices of $\tau \in T$. Thus, there are at most $200/\eps^2$ choices for which $x$ is not $\tau$-stable.
\end{proof}

\subsubsection*{Colors}

\begin{definition}[Window colors]\label{def:window-color}

For a given node $\cI^i$ of the algorithm tree,
we partition the windows in $\cI_{good}^i$ into a constant number of {\bf window-colors}, or equivalence classes.
Specifically, we say that $y,z \in \cI_{good}^i$ with $y, z \in \cA$ or $y, z \in \cB$ belong to the same color, denoted as $z \in \phiw(y)$ (equivalently $y \in \phiw(z)$), if all of the following hold, for every $\tau \in T$, every interval-multiplier $\lambda \in \Lambda$, and every pair of intervals $I_y,I_z \in \cL^i$ such that $\lambda I_y \ni y$, $\lambda I_z \ni z$, and $\lambda I_y \equiv \lambda I_z$.

\begin{description}
\item[global density] $C^{i+1}(y)$ and $C^{i+1}(z)$ have approximately\footnote{Recall that all terms in the definition of density are snapped to the nearest power of $t^{\eps_0}$.} the same global density, including when intersected with $\cI^i_{good}$:
  \begin{align}\label{eq:color-global}
    \rho^{\textrm{global}}(C^{i+1}(y, \tau)) &= \rho^{\textrm{global}}(C^{i+1}(z, \tau))
  \end{align}
  
\item[local density] $C^{i+1}(y, \tau)$ and $C^{i+1}(z, \tau)$ have approximately the same local density. Namely,

  \begin{align}\label{eq:color-degree}
    \rho^{\textrm{local}}_{\lambda I_y}(C^{i+1}(y, \tau)) &= \rho^{\textrm{local}}_{\lambda I_z}(C^{i+1}(z, \tau))
  \end{align}

\end{description}
\end{definition}

\begin{definition}[Matching colors]\label{def:color}
For all $y \in \cI^i_{good}$, we say that $y$'s {\bf matching-color} (or just {\bf color}) is
  \[
    \phim(y) := \{z \in \cI^i_{good} : \phiw(y) = \phiw(z) \text{\;AND\;} \phiw(\mu(y)) = \phiw(\mu(z))\}. 
  \]
\end{definition}

We let $\Phi^i$ denote the partition of $\cI^i_{good}$ into matching-colors (which are just called colors in the rest of the paper). For a given $y \in \cI^i_{good}$, $\phim(y) \in \Phi^i$ is $y$'s color. Note that $y$ and $z$ have the same color if and only if $y \in \phim(z)$ and $z \in \phim(y)$. Also note that $y$ and $z$ have the same color if and only if $\mu(y)$ and $\mu(z)$ have the same color (this fact is critical in Claim~\ref{claim:size-hat-W}).

For the analysis, note that $|\Phi^i| \le (20/\eps)^{2/\eps}$ which is a constant independent of $t$. In particular $t^{\eps} \ge |\Phi^i|$ as long as $t \ge t_{\min}$. Also note that color is independent of the algorithm's choice of $x$ and $\tau$.

The following are some of the most useful facts about colors. We start with showing that the property of being $\tau$-stable is consistent within a color. This is an important tool in proofs of later propositions. 

\begin{proposition}[Stability is identical within color.]\label{prop:color-stab}
  For all $y \in \cI^i_{good}$, $z \in \phim(y)$, and $\tau \in T$, $y$ is $\tau$-stable if and only if $z$ is $\tau$-stable.
\end{proposition}

\begin{proof}
  It suffices to show that $y$ is $\tau$-stable implies that $z$ is $\tau$-stable.
  
  Let $\lambda I_z \in \lambda \cL$ be any interval such that $z \in \lambda I_z$. We have that there is an equivalent interval $\lambda I_y \in \lambda \cL$ such that $y \in \lambda I_y$.
  \[
    \lrho_{\lambda I_z}(C^{i+1}(z, \tau-2)) = \lrho_{\lambda I_y}(C^{i+1}(y, \tau-2)) = \cdots = \lrho_{\lambda I_y}(C^{i+1}(z, \tau+2)) = \lrho_{\lambda I_z}(C^{i+1}(z, \tau+2)).
  \]

  By a similar argument, the global densities are also stable. Thus, $z$ is $\tau$-stable.
\end{proof}

Next, we show that most windows are in a large color. This makes sure that when we randomly sample windows in the algorithm, we are likely to sample a distribution of colors which reflects the true distribution.

\begin{proposition}[Most $y$ are in a large color.]\label{prop:most-in-big-color}
  Assume that $\cI^i$ is an active node. For all but $t^{-2\eps_i}$ fraction of $y \in \cI^i_{good}$,
  \begin{gather}\label{eq:color-1}
    \frac{|\phim(y)|}{|\cI^i|} \ge t^{-4\eps_i}.
  \end{gather}
\end{proposition}

\begin{proof}
  Let $\Phi'^i \subset \Phi^i$ be the set of colors $\phi \in \Phi^i$ such that
  \[
    |\phi| \ge \frac{t^{-2\eps_i}}{|\Phi^i|} |\cI^i_{good}|.
  \]
  For all $y \in \phi \in \Phi'^i$, we have that
  \begin{align*}
    \frac{|\phim(y)|}{|\cI^i|} &\ge \frac{t^{-2\eps_i}}{|\Phi^i|}\cdot \frac{|\cI^i_{good}|}{|\cI^i|}\\
                                   &\ge t^{-4\eps_i} &&\text{ by Eq.~\eqref{eq:invariant-good-big} and $|\Phi^i| \le t^{\eps_i}$.}
  \end{align*}
  Observe that
  \begin{align*}
    |\{y \in \phi \in \Phi'^i\}| &= \sum_{\phi \in \Phi'^i} |\phi|\\
                                 &\ge \sum_{\phi \in \Phi^i} |\phi| - |\Phi^i| \frac{t^{-\eps_i}}{|\Phi^i|} |\cI^i_{good}|\\
                                 &= (1 - t^{-2\eps_i}) |\cI^i_{good}|,
  \end{align*}
  as desired.
\end{proof}

We now prove a stronger version that most intervals around most windows have many of the same color. That is, most windows are not ``isolated.'' This helps makes sure that most windows are `discoverable' by making it likely that for most windows, some sampled window will be of the same color \emph{and} close in proximity. 

\begin{proposition}[In $y$'s interval, there are many $z$ of the same color.]\label{prop:most-in-interval-in-big-color}
  Assume that $\cI^i$ is active. For $1 - t^{-2\eps_i}$ fraction of $y \in \cI^i_{good}$ for all $I \in \cL^i$ such that $y \in I$,
  \begin{gather}\label{eq:color-2prime}
    \frac{|\phim(y) \cap I|}{|\cI^i \cap 49 I|} \ge t^{-4\eps_i}.
  \end{gather}
\end{proposition}

\begin{proof}
  Let $B_{I}$ be the set of $y \in \cI^i_{good} \cap I$ such that
  \[
    \frac{|\phim(y) \cap I|}{|\cI^i \cap 49 I|} < t^{-4\eps_i}.
  \]
  By logic similar to that of the proof of Proposition~\ref{prop:most-in-big-color}, we have that
  \[
    |B_{I}| \le t^{-4\eps_i} |\Phi^i| |\cI^i \cap 49 I|.
  \]
  Summing over all choices of $I$, we have that
  \begin{align*}
    \sum_{I} |B_{I}| &\le t^{-4\eps_i} |\Phi^i| \sum_{I} |\cI^i \cap 49 I|\\
                                               &\le t^{-4\eps_i} 49(2/\eps)|\Phi^i| |\cI^i|\\
                                               &\le t^{-2\eps_i} |\cI^i_{good}| &&\text{ by Eq.~\eqref{eq:invariant-good-big} and $|\Phi^i|(100/\eps) \le t^{\eps}$}.
  \end{align*}
  Thus, at least $1 - t^{-2\eps_i}$ fraction of $y \in \cI^i_{good}$ satisfy Eq.~\eqref{eq:color-2prime} for all choices of $I$.
\end{proof}

The following is a refined version of the above proposition.

\begin{corollary}\label{cor:most-in-interval-in-big-color}
  Assume that $\cI^i$ is active. For $1 - t^{-2\eps_i}$ fraction of $y \in \cI^i_{good}$ for all $\lambda I \in \lambda\cL^i$ such that $y \in \lambda I$,
  \begin{gather}\label{eq:color-2}
    \frac{|\phim(y) \cap \lambda I|}{|\cI^i \cap \lambda I|} \ge t^{-4\eps_i}.
  \end{gather}
\end{corollary}

\begin{proof}
  Consider $I_y \in \cL^i$ equivalent to $I$ such that $y \in I_y \subset \lambda I$. Then, $\lambda I \subset 49 I_y$. Thus, whenever \eqref{eq:color-2prime} is satisfied
  \[
    \frac{|\phim(y) \cap \lambda I|}{|\cI^i \cap \lambda I|} \ge \frac{|\phim(y) \cap I_y|}{|\cI^i \cap 49 I_y|}\ge t^{-4\eps_i}.
  \]
\end{proof}

\subsubsection*{Maximal intervals}

In the query algorithm, we take the union of `dense' intervals around each of the sampled windows. These intervals are formally defined as follows.

\begin{definition}[Maximal interval]\label{def:maximal-interval}
  For any $x, y \in \cI^i$ and $\tau \in \Tau$, define $J(y, x, \tau)$ to be the maximal $J \in \cL$ such that $y \in J$ and
  \begin{gather}\label{eq:maximal-interval}
    \rrho_{7J}(C^{i+1}(x, \tau)) \ge \rho.
  \end{gather}
  If no such interval exists, then $J(y, x, \tau) = \perp$
\end{definition}

Note that if $x$ and $\tau$ are chosen by the algorithm and $J(y, x, \tau) \neq \perp$, then $7J(y, x, \tau) \cap \cI^i \subset \cI^{i+1}$.
 In particular, $y \in \cI^{i+1}$.

We also define
\[
  \cJ(x, \tau) := \{J(y, x, \tau) \cap \cI^i : y \in \cI^i, J(y, x, \tau) \neq \perp\}.
\]

Note that $\bigcup_{J \in \cJ(x, \tau)}7J = \cI^{i+1}$.  First, we show that every maximal interval is very close in relative density to the `target' relative density of $\rho$.

\begin{proposition}\label{prop:maximal-density-close}
  For all $J \in \cJ(x, \tau)$ for which $7J \cap \cI^i_{dd} \neq \emptyset$,
  \[
    \rrho_{7J}(C^{i+1}(x, \tau)) \in [\rho, t^{4\eps_i}\rho].
  \]
\end{proposition}

\begin{proof}
  The lower bound follows immediately from the definition of $\cJ(x, \tau)$.

  To show the upper bound, first note that $J$ cannot be the entire string ($\cA$ or $\cB$), since then
  \[
    \rrho_{7J}(C^{i+1}) = 1 \ll \rho.
  \]
  Otherwise, we can consider $J' \subsetneq J$ which is the next larger interval in $\cL$. Since $J$ is a maximal interval such that Eq.~\eqref{eq:maximal-interval} holds,
  \begin{align*}
    \rho &> \rrho_{7J'}(C^{i+1}(x, \tau))\\
         &= \rrho_{7J}(C^{i+1}(x, \tau)) \frac{\rrho_{7J'}(C^{i+1}(x, \tau))}{\rrho_{7J}(C^{i+1}(x, \tau))}\\
         &= \rrho_{7J}(C^{i+1}(x, \tau)) \frac{\lrho_{7J'}(C^{i+1}(x, \tau))}{\lrho_{7J}(C^{i+1}(x, \tau))}\\
         &= \rrho_{7J}(C^{i+1}(x, \tau)) \frac{\snap |C^{i+1}(x, \tau) \cap 7J'|}{\snap |C^{i+1}(x, \tau) \cap 7J|} \cdot \frac{\snap |\cI^i \cap 7J|}{\snap |\cI^i \cap 7J'|}\\
         &\ge  t^{-6\eps}\rrho_{7J}(C^{i+1}(x, \tau)) \frac{|7J'|}{|7J|} \cdot \frac{\eta_{7J}(\cI^i)}{\eta_{7J'}(\cI^i)}\\
         &\ge t^{-4\eps_i} \rrho_{7J}(C^{i+1}(x, \tau)),
  \end{align*}
  where the last line uses that $7J \cap \cI^i_{dd} \neq \emptyset$.
\end{proof}

The following proposition unpacks the previous proposition into an explicit statement about interval lengths and their intersections with $C^{i+1}$.

\begin{proposition}\label{prop:interval-to-clique-ratio}
  For all $J \in \cJ$ such that $7J \cap \cI^i_{dd} \neq \emptyset$
  \begin{gather}\label{eq:interval-to-clique-ratio}
    \frac{1}{\rho}\cdot \frac{|C^{i+1} \cap 7J|}{|C^{i+1}|} \cdot \frac{|\cI^i|}{|\cI^i \cap 7J|} \in [t^{-5\eps_i}, t^{5\eps_i}].
  \end{gather}
\end{proposition}

\begin{proof}
  \begin{align}\label{eq:i-i-ratio-1}
    |\cI^{i} \cap 7 J| &\approx_{t^{2\eps)}} \frac{|C^{i+1} \cap 7 J|}{\lrho_{7J}(C^{i+1})} \nonumber\\
                               &\approx_{t^{2\eps}} \frac{|C^{i+1} \cap 7 J| }{\rrho_{7J}(C^{i+1}) \cdot \grho(C^{i+1})} \nonumber\\
                               & \approx_{t^{2\eps + 4\eps_i}} \frac{|C^{i+1} \cap 7 J|} {\rho \cdot \grho(C^{i+1})}\nonumber\\
                               &\approx_{t^{4\eps + 4\eps_i}} \frac{|C^{i+1} \cap 7 J|\cdot|\cI^{i}|} {\rho \cdot |C^{i+1}|}.
  \end{align}
  Most of the steps follow by definitions of local, global, and relative densities (with approximation due to the snaps), with the exception of the third line, which follows from Proposition~\ref{prop:maximal-density-close}. Note that $t^{4\eps + 4\eps_i} \le t^{5\eps_i}$ because $t \ge t_{\min}$. 
\end{proof}

We also need to make sure that the maximal intervals are not $\bot$.

\begin{proposition}\label{prop:good-windows-remain}
  For all $x \in \cI^i$ and $\tau \in \Tau$ such that $\grho(C^{i+1}(x, \tau)) \le \rho^{-1}$. For all $y \in C^{i+1}(x, \tau)$, we have that
  \[
    J(y, x, \tau) \neq \bot.
  \]
  In particular, for all $y \in \cI^i$, such that $C^{i+1}(y, \tau)$ is small, $J(y, y, \tau) \neq \bot$.
\end{proposition}

\begin{proof}
  Consider the interval $7\{y\}$ around $y$. Then,
  \[
    \rrho_{7\{y\}}(C^{i+1}) \ge \frac{1}{\snap(7) \grho(C^{i+1})} = \frac{1}{\grho(C^{i+1})}\ge \rho. 
  \]
  Thus, either $\{y\}$ or a superset of $\{y\}$ is in $\cJ(x, \tau)$.
\end{proof}

We also show that two windows of the same color in closer proximity have the same maximal interval.

\begin{proposition}\label{prop:J-same}
  Fix $\tau \in T$.  For all $y \in \cI^i_{good}$ such that $C^{i+1}(y, \tau)$ is small, and $z \in \phim(y) \cap J(y, y, \tau)$, we have that
  \[
    J(y, y, \tau) = J(z, z, \tau).
  \]
\end{proposition}

\begin{proof}
  For brevity, let $J_y := J(y, y)$ and $J_z = J(z, z)$. Since $z \in J_y$ and $y$ and $z$ have the same color,
  \[
    \rrho_{7 J_y}(C^{i+1}(z)) = \rrho_{7 J_y}(C^{i+1}(y)) \ge \rho^i.
  \]
  Thus, $J_y \subset J_z$ (by maximality of $J_z$). Which in particular implies that $y \in J_z \cap \phim(y)$. Therefore,
  \[
    \rrho_{7 J_z}(C^{i+1}(y)) = \rrho_{7 J_z}(C^{i+1}(z)) \ge \rho^i.
  \]
  Therefore, $J_z \subset J_y$, and so $J_y = J_z$.
\end{proof}

In addition, we show that these maximal intervals do not change at $\tau$-stable points when the radius is changed.

\begin{proposition}[$J$ stability]\label{prop:J-stable}
  For a fixed $\lambda$ and $\rho$, for all $y \in \cI^{i}_{good}$, if $y$ is $\tau$-stable (Definition \ref{def:tau-stable}), then
  \[
    J(y, y, \tau-2) = J(y, y, \tau-1) = J(y, y, \tau) = J(y, y, \tau+1) = J(y, y, \tau+2).
  \]
  Thus, for all $z \in \phim(y) \cap J(y, y, \tau)$, by Proposition~\ref{prop:J-same},
  \[
    J(z, z, \tau-2) = \cdots = J(z, z, \tau+2) = J(y, y, \tau).
  \]
\end{proposition}

\begin{proof}
  Since $y$ is $\tau$-stable and $y \in \lambda J(y, y, \tau)$, we have that for all $j \in \{-2, -1, 0, 1, 2\}$.
  \begin{align*}
    \rrho_{\lambda J(y, y, \tau)}(C^{i+1}(y, \tau+j)) &= \frac{\snap |C^{i+1}(y, \tau+j) \cap \lambda J(y, y, \tau)|}{\snap |\cI^i \cap \lambda J(y, y, \tau)|} \cdot \frac{\snap |\cI^i|}{\snap |C^{i+1}(y, \tau+j)|}\\
                                                                     &= \frac{\snap |C^{i+1}(y, \tau) \cap \lambda J(y, y, \tau)|}{\snap |\cI^i \cap \lambda J(y, y, \tau)|} \cdot \frac{\snap |\cI^i|}{\snap |C^{i+1}(y, \tau)|}\\
                                                                     &= \rrho_{\lambda J(y, y, \tau)}(C^{i+1}(y, \tau)) \ge \rho^i.
  \end{align*}
  Thus, $J(y, y, \tau) \subset J(y, y, \tau+j)$ by maximality. By a nearly identical argument (c.f., Proposition~\ref{prop:J-same}), we have that $J(y, y, \tau+j) \subset J(y, y, \tau)$, so $J(y, y, \tau+j) = J(y, y, \tau)$.
\end{proof}

We now show that going from $\cI^i$ to $\cI^{i+1}$ in the algorithm tree decreases the size of $\cI^i$ by a factor of roughly $\rho$. This makes sure that the total work done in each level of the algorithmic tree is approximately the same.

\begin{proposition}\label{prop:next-I-small}
  Consider a node $\cI^i$ at level $i \in \{0, 1, \hdots, i_{\max} - 1\}$ of the algorithm tree. Assume that $x, \tau$ are such that $\grho(C^{i+1}(x, \tau)) \le \rho^{-1}$ (the small clique case), then 
  \begin{gather}\label{eq:cI-small-relative}
    |\cI^{i+1}(x, \tau)| \le t^{5\eps}\rho^{-1}|\cI^i|.
  \end{gather}
  In particular, for all $i \in \{0, 1, \hdots, i_{\max}\}$,
  \begin{gather}\label{eq:cI-small-absolute}
    |\cI^i| \le t^{1+5i\eps}\rho^{-i}.
  \end{gather}
\end{proposition}

\begin{proof}
  By definition,
  \begin{align*}
    |\cI^{i+1}(x, \tau)| &\le \sum_{J \in \cJ(x, \tau)} |\cI^i \cap 7 J|\\
                         &\le t^{2\eps} \sum_{J \in \cJ(x, \tau)} \frac{|C^{i+1} \cap 7J|}{\lrho_{7J}(C^{i+1})}\\
                         &= t^{2\eps} \sum_{J \in \cJ(x, \tau)} \frac{|C^{i+1} \cap 7J|}{\snap |C^{i+1}|}\cdot \frac{\snap |C^{i+1}|}{\snap |\cI^i|} \frac{\snap |\cI^i|}{\lrho_{7J}(C^{i+1})}\\
                         &\le t^{3\eps} |\cI^i| \sum_{J \in \cJ(x, \tau)} \frac{|C^{i+1} \cap 7J|}{\snap |C^{i+1}|} \cdot \frac{1}{\rrho_{7J}(C^{i+1})}&&\text{(Definition of $\rrho$)}\\
                         &\le t^{4\eps} \rho^{-1} |\cI^i| \sum_{J \in \cJ(x, \tau)} \frac{|C^{i+1} \cap 7J|}{|C^{i+1}|}&&\text{(Definition of $J$)}\\
                         &\le t^{5\eps} \rho^{-1} |\cI^i|,
  \end{align*}
  where the last line follows from the fact that every $y \in C^{i+1}$ is in $\ll t^{\eps}$ intervals. The second statement follows by induction, and that $|\cI^0| \le t$.
\end{proof}

\subsubsection*{More on interval multipliers}

Also note the following fact about interval multipliers, which will be of use during the analysis. Essentially it says that if two intervals share an edge of $\mu$, then $\mu$ fully sends one into the scaled interval of the other. This proposition is the key reason we assume that $D(\mu) \le 2$.

\begin{claim}\label{claim:matching-inclusion}
  Let $I_1, I_2 \in \cL$ be intervals such that there exists $y \in I_1$ such that $\mu(y) \in I_2$. Then, at least one of the following is true
  \begin{align*}
    \mu(I_1) &\subset 7I_2\\
    \mu(I_2) &\subset 7I_1
  \end{align*}
\end{claim}

\begin{proof}
  Let $z_1$ be the first window in $I_1$ such that $\mu(z_1) \neq \bot$. Let $z'_1$ be the last window in $I_1$ such that $\mu(z_1) \neq \bot$. Define $z_2, z'_2 \in I_2$ analogously.

  It suffices to show that either $\mu(z_1), \mu(z'_1) \in 7I_2$ or $\mu(z_2), \mu(z'_2) \in 7I_1$. Assume for sake of contradiction that both are false. In particular, this means that at least one of $\mu(z'_2)$ is $\mu(z_2)$ is at least three times the length of the interval $I_1$ away from $y$ (respective fact for $I_2$, too).

  In equations,
  \begin{align*}
    |s(\mu(z'_2)) - s(\mu(z_2))| &\ge 3|s(z'_1) - s(z_1)|\\
    |s(\mu(z'_1)) - s(\mu(z_1))| &\ge 3|s(z'_2) - s(z_2)|
  \end{align*}
  Now, using that $\mu$ has skew at most $2$, we have
  \begin{align*}
    |s(z'_1) - s(z_1)| &\ge \frac{1}{2}|s(\mu(z'_1)) - s(\mu(z_1))|\\
                       &\ge \frac{3}{2}|s(z'_2) - s(z_2)|\\
                       &\ge \frac{3}{4}|s(\mu(z'_2)) - s(\mu(z_2))|\\
                       &\ge \frac{9}{4}|s(z'_1) - s(z_1)|,
  \end{align*}
  a contradiction unless $s(z'_1) = s(z_1)$ in which case $\mu(I_1) = \{\mu(y)\} \subset 7I_2$. Thus, the claim holds.
\end{proof}

\subsection{Correctness of the query algorithm}\label{subsec:query-correct}

\subsubsection{Surviving windows}\label{subsec:alg-invariant}

\begin{definition}[Surviving window]\label{def:survive}
Given a node $\cI^i$ at level $i$ in the algorithm tree, we say that $y \in \cI^i_{good}$ \emph{survives} if the edge $(y, \mu(y))$ is discovered with probability at least $1 - e^{-t^{\eps_i}}$ over the remaining randomness in the algorithm and assuming all calls to $\textsc{MAIN}(\cdot, L-1)$ are accurate.
\end{definition}

The following is the key lemma we need concerning active nodes--that almost every active node survives in at least one branch of the algorithm tree.

\begin{lemma}[Main inductive step]\label{lem:inductive}\hfill\\
  Assume that $\cI^i$ is an active node at level $i \in \{0, \hdots, i_{\max}\}$ in the algorithm tree. Then, all but $\delta_i := 2t^{-\eps_i}$-fraction of $y \in \cI^i \cap \mu(\cI^i)$ survive.
\end{lemma}

Before we prove Lemma~\ref{lem:inductive}, we show the following necessary ingredient which shows that most active nodes have some of the properties for being active in the next level.

\begin{lemma}[The Color Lemma]\label{lem:color}
  Let $\cI^i$ be any active node at level $i \in \{0, 1, \hdots, i_{\max} - 1\}$ in the algorithm tree, then for all but $2t^{-2\eps_i}|\cI^i_{good}|$ windows $y \in \cI^i_{good}$ with probability at least $t^{-12\eps_i}\rho^{-1}$ over the choice of $x$ and $\tau$, we have that either
  \begin{itemize}
  \item $y, \mu(y) \in C^{i+1}(x, \tau)$, or
  \item $y, \mu(y) \in \cI^{i+1}(x, \tau)$.
  \end{itemize}
\end{lemma}

Note that this Lemma is a bit weaker than Lemma~\ref{lem:inductive}, as we only obtain that $y, \mu(y) \in \cI^{i+1}$ (instead of $\cI^{i+1}_{good}$) and we do not know whether $\cI^{i+1}$ is active. We call this result the Color Lemma as the definition of colors (Definition~\ref{def:color}) is crucial to its proof.

\subsubsection{Proof of the Color Lemma}\label{sub:proof-color}

To prove the Color Lemma (Lemma~\ref{lem:color}), we need for every $y, \mu(y)$ pair to identify a large collection of potential $x$ such that either $y, \mu(y) \in C^{i+1}(x, \tau)$ or $y, \mu(y) \in \cI^{i+1}(x, \tau)$. The first case is very likely if and only if $C^{i+1}(y, \tau)$ is a big clique. Otherwise, most $x$ come from a different set.

Let $\cI^i_{big} \subset \cI^i_{good}$ be the set of nodes for which $C^{i+1}(y, \tau_{\max}-1)$ is big. Let $\cI^i_{small} := \cI^i_{good} \setminus \cI^i_{big}$.

Consider any $y \in \cI^i_{small}$. For all $\tau, \tau' \in \Tau$ and $\lambda \in \Lambda$ define
\begin{gather}\label{eq:def-W}
  W_\lambda(y, \tau, \tau') := \bigcup_{z \in \phim(y) \cap \lambda J(y, y, \tau)} C^{i+1}(z, \tau').
\end{gather}

That is $W$ is the union of $y$'s clique and all cliques whose center $z$ is nearby and has the same color as $y$.

We also define a variant of $W$ which is symmetric in $y$ and $\mu(y)$:
\begin{gather}\label{eq:def-hat-W}
  \widehat{W}(y, \tau) = W_{7}(y, \tau, \tau-1) \cap W_{7}(\mu(y), \tau, \tau-1).
\end{gather}

Our goal now is to show that for any $w \in \widehat{W}$, $y, \mu(y) \in \cI^{i+1}$. To do this, we need to show some maximal intervals will contains $y$ and $\mu(y)$.

\begin{claim}[$w$ finds $y$]\label{claim:density-W}
  For all $y \in \cI^i_{good}$ either $y \in \cI^i_{big}$ or for any $\tau \in \Tau$ such that $y$ is $\tau$-stable and for all $w \in W_{7}(y, \tau, \tau-1)$,
  \begin{align}
    \rrho_{7 J(y, y, \tau)}(C^{i+1}(w, \tau)) &\ge \rho\label{eq:density-W}
  \end{align}
  In particular, $J(y, y, \tau) \subset J(y, w, \tau) \neq \bot$, so $y \in \cI^{i+1}(w, \tau)$.
\end{claim}

\begin{proof}
  Let $W := W_{7}(y, \tau, \tau-1)$. Let $J(\tau) := J(y, y, \tau)$. Assume that $y$ is stable at $\tau$. By Proposition~\ref{prop:J-stable}, $J(\tau) = J(\tau-1) = J(\tau+1)$; let $J$ be this common interval.

  For any $z \in \phim(y) \cap 7J$, we have that $z$ is stable at $\tau$ by Proposition~\ref{prop:color-stab}, and since $z, y \in 7J$,
  \[
    \rrho_{7J}(C^{i+1}(z, \tau)) = \rrho_{7J}(C^{i+1}(y, \tau)) \ge \rho.
  \]
  Thus, it suffices to show that $\rho_{7J}(C^{i+1}(w, \tau)) = \rho_{7J}(C^{i+1}(z, \tau))$  for all $w \in C^{i+1}(z, \tau-1)$. Notice by Proposition~\ref{prop:clique-properties},
  \begin{align*}
    C^{i+1}(z,\tau-1) \subseteq C^{i+1}(z,\tau), &C^{i+1}(w,\tau) \subseteq C^{i+1}(z,\tau+1)\\
    C^{i+1}(z,\tau-1) \cap 7J \subseteq C^{i+1}(z,\tau) \cap 7J, &C^{i+1}(w,\tau) \cap 7J \subseteq C^{i+1}(z,\tau+1) \cap 7J.\\
  \end{align*}

  Since $z$ is stable at $\tau$, this means that
  \begin{align*}
    \grho(C^{i+1}(z, \tau-1)) &= \grho(C^{i+1}(w, \tau)) = \grho(C^{i+1}(z, \tau+1))\\
    \lrho_{7J}(C^{i+1}(z, \tau-1)) &= \lrho_{7J}(C^{i+1}(w, \tau)) = \lrho_{7J}(C^{i+1}(z, \tau+1)).
  \end{align*}
  Thus, $\rrho_{7J}(C^{i+1}(w, \tau)) = \rrho_{7J}(C^{i+1}(z, \tau)) \ge \rho$, as desired.
\end{proof}

Further, we need to make sure that $W$ is large so that we have a large chance of sampling a window from $W$ (and thus the window $y$ survives on). We make exception to the case $y$ is in a big clique, since that is likely to be discovered.

\begin{claim}[$W$ is big]\label{claim:size-W}
  For $1 - t^{-2\eps_i}$ fraction of $y \in \cI^i_{good}$ either $y \in \cI^i_{big}$ or for all $\tau \in \Tau$ such that $y$ is $\tau$-stable.
  \begin{gather}\label{eq:color-5}
    |W_1(y, \tau, \tau-2)| \ge \frac{t^{-11\eps_i}}{\rho} |\cI^i|.
  \end{gather}
\end{claim}

\begin{proof}  
  Assume that $y \in \cI_{small}$ and such that Proposition~\ref{prop:most-in-interval-in-big-color} holds. Since $y$ is $\tau$-stable, Proposition~\ref{prop:J-stable} holds for $y$.  Fix $J := J(y, y, \tau)$.

    Let $W := W_1(y, \tau, \tau-2).$ The size of $W$ is lower bounded by the sum of the set sizes, divided by the maximum number of $z$'s that could have generated any given $w$:
  \begin{gather}\label{eq:color-6}
    |W| \geq \frac{\sum_{z \in \phim(y) \cap J}\left|C^{i+1}(z,\tau-2) \right|}{\max_{w \in W} \Big|C^{i+1}(w,\tau-2) \cap \phim(y) \cap J\Big|}.
  \end{gather}
  Note that if $z \in C^{i+1}(w, \tau-2)$, then $C^{i+1}(w, \tau-2) \subset C^{i+1}(z, \tau)$ (by Proposition~\ref{prop:clique-properties}). Thus, for a fixed pair of $w$ and $z$,
\begin{align*}
  |C^{i+1}(w,\tau-2) \cap \phim(y) \cap J| &\le |C^{i+1}(z,\tau) \cap \phim(y) \cap J|\\
                                               &\le |C^{i+1}(z,\tau) \cap 7 J|.
\end{align*}

Because $J \in \cJ(z, \tau)$ (by Proposition~\ref{prop:J-stable}), we have by Proposition~\ref{prop:interval-to-clique-ratio} that
\begin{align*}
  |C^{i+1}(z, \tau) \cap 7 J| &\le t^{5\eps_i} \rho \cdot \frac{|C^{i+1}(z, \tau)|}{|\cI^i|}\cdot |\cI^i \cap 7 J|\\
 & \le t^{6\eps_i} \rho \cdot \frac{|C^{i+1}(z, \tau-2)|}{|\cI^i|}\cdot |\cI^i \cap 7 J|&&\text{(Stability)}.
\end{align*}

Plugging back into~\eqref{eq:color-6}, we have that
\begin{align*}
  |W| & \ge t^{-6\eps_i} \frac{\sum_{z \in \phim(y) \cap J}\left|C^{i+1}(z,\tau-2) \right|}{\rho \cdot \max_{z' \in \phim(y)\cap J} \left(\big|C^{i+1}(z',\tau-2))\big|/|\cI^i|\right) \cdot |\cI^i \cap 7 J|} \\
      &\ge  t^{-7\eps_i} \frac{\sum_{z \in \phim(y) \cap J}\left|C^{i+1}(z,\tau-2) \right|}{\min_{z \in \phim(y) \cap J}\left|C^{i+1}(z,\tau-2) \right|}\frac{|\cI^i|}{\rho|\cI^i \cap 7 J|}&&\text{(definition of color)}\\
      &\ge t^{-7\eps_i}\frac{|\phim(y) \cap J||\cI^i|}{\rho|\cI^i \cap 7 J|}\\
& \geq \frac{t^{-11\eps_{i}}}{\rho}|\cI^i|, && \text{(using~\eqref{eq:color-2})}
\end{align*}
as desired.
\end{proof}

The previous statements were just that $y$ survives. Now we need to show that $y$ and $\hat{y}$ survive.

\begin{claim}[$\widehat{W}$ is big]\label{claim:size-hat-W}
  For $1 - 2t^{-2\eps_i}$ fraction of $y \in \cI^i_{good}$ either $y \in \cI^i_{big}$ or $\mu(y) \in \cI^i_{big}$ or for all $\tau \in \Tau$ such that both $y$ and $\mu(y)$ are $\tau$-stable,
  \begin{gather}\label{eq:color-7}
    |\widehat{W}(y, \tau)| \ge \frac{t^{-11\eps_i}}{\rho} |\cI^i|.
  \end{gather}
\end{claim}

\begin{proof}
  Assume that neither $y, \mu(y) \in \cI^i_{smal}$.
  Since at most  $t^{-2\eps_i}$ fraction of $y \in \cI^i_{good}$ fail to satisfy Claim~\ref{claim:size-W}, for all $\tau$ such that $y$ is $\tau$-stable, we must have at most  $2t^{-2\eps_i}$ fraction of $y \in \cI^i_{good}$ fail to satisfy Claim~\ref{claim:size-W} for at least one of $y$ and $\mu(y)$ (whenever both $y$ and $\mu(y)$ are $\tau$-stable).

  Consider any $y$ and $\tau$ such that both $y$ and $\mu(y)$ are $\tau$-stable and Claim~\ref{claim:size-W} is satisfied for both $y$ and $\mu(y)$. Since $\widehat{W}(y, \tau) = W_7(y, \tau, \tau-1) \cap W_7(\mu(y), \tau, \tau-1)$, it suffices then to show that either
  \[
    W_1(y, \tau, \tau-2) \subset W_7(\mu(y), \tau, \tau-1)\text{ or }W_1(\mu(y), \tau, \tau-2) \subset W_7(y, \tau, \tau-1).
  \]
  Let $J_y := J(y, y, \tau)$ and $J_{\mu(y)} := J(\mu(y), \mu(y), \tau)$. By Claim~\ref{claim:matching-inclusion}, we must then have that either
  \[
    \mu(J_y) \subset 7J_{\mu(y)}\text{ or }\mu(J_{\mu(y)}) \subset 7J_{y}.
  \]
  If the first case holds, for all $z \in \phim(y) \cap J_y$,
  \[
    C^{i+1}(z, \tau-2) \subset C^{i+1}(\mu(z), \tau-1) \subset W_7(\mu(y), \tau, \tau-1).
  \]
  Thus, $W_1(y, \tau, \tau-2) \subset \widehat{W}(y, \tau)$. Likewise, if the second case holds, for all $z' \in \phim(\mu(y)) \cap J_{\mu(y)}$
  \[
    C^{i+1}(z', \tau-2) \subset C^{i+1}(\mu(z'), \tau-1) \subset W_7(y, \tau, \tau-1).
  \]
  Then, $W_1(\mu(y), \tau, \tau-2) \subset \widehat{W}(y, \tau)$. Note we are using the fact that if $y$ and $z$ have the same color, then $\mu(y)$ and $\mu(z)$ have the same color.
\end{proof}

Now we have all the tools we need to prove the color lemma.

\begin{proof}[Proof of the Color Lemma.]
  First, consider the case that $y \in \cI^i_{big}$. That is, $\grho(C^{i+1}(y, \tau_{\max}-1)) > \rho^{-1}$. Then, by triangle inequality, for all $x \in C^{i+1}(y, \tau_{\max}-1)$, the clique $C^{i+1}(x, \tau_{\max})$ contains both $y$ and $\mu(y)$. The probability of sampling $x \in C^{i+1}(y, \tau_{\max} - 1)$ and $\tau = \tau_{\max}$ is at least \[\frac{|C^{i+1}(y, \tau_{\max}-1)|}{|\cI^i|\tau_{\max}} \ge t^{-3\eps}\frac{1}{\grho(C^{i+1}y, \tau_{\max}-1)} \ge \frac{t^{-11\eps_i}}{\rho}.\]
  Similar logic covers the case $\mu(y) \in \cI^i_{big}$.

  Otherwise, we may assume that $\grho(C^{i+1}(y, \tau)) \le \rho$ for all $\tau < \tau_{\max}$. For all but $2t^{-2\eps_i}|\cI^i_{good}|$ such $y$, we have by Claim~\ref{claim:density-W}  that $y, \mu(y) \in \cI^{i+1}(x, \tau)$ for all $x \in \widehat{W}(y, \tau)$ and any $\tau \in \Tau$ such that $y$ and $\mu(y)$ are $\tau$ stable. Further, by Claim~\ref{claim:size-hat-W}, we have for all such $y$ and $\tau$ that $|\widehat{W}(y, \tau)| \ge \frac{t^{-11\eps_i}}{\rho} |\cI^i|$. Since the probability of selecting a proper $\tau$ is at least $\frac{\tau_{\max} - 400/\eps^2}{\tau_{\max}} \ge t^{-\eps_i}$ by Claim~\ref{claim:tau-stable}, the probability of success again is at least $\frac{t^{-12\eps_i}}{\rho}$.

  Note that for all $x \in \widehat{W}(y, \tau)$, $C^{i+1}(x, \tau)$ is not big, as we proved that there exists an interval for which the clique has relative density $\rho$, implying the global density is at most $\rho^{-1}$.
\end{proof}

\subsubsection{Proof of the main inductive step (Lemma~\ref{lem:inductive})} \label{sub:proof-inductive}

\begin{proof}
  We prove this lemma by ``reverse'' induction with $i = i_{\max}$ as the base case. Note that for every node at level $i_{\max}$, all pairs of windows in $\cI^{i_{\max}}$ are queried. Thus, the base case follows as $\delta_{i_{\max}} \ge 0$. 
  
  Now assume that the statement is true for some $i+1$, we seek to show that it is also true for $i$. Let $\cI^i$ be any active node at level $i$.

  Create a bipartite graph with $Y := \cI^i \cap \mu(\cI^i)$ and $W := \cI^i \times T$. Let $E_1 \subset Y \times W$ be all pairs $(y, (w, \tau))$ such that
  \[
    y,\mu(y) \in C^{i+1}(w, \tau)\text{ or } y,\mu(y) \in \cI^{i+1}(w, \tau).
  \]
  Let $Y' = \cI^i_{good} = \cI^i_{dd} \cap \mu(\cI^i_{dd})$. Since $\cI^i$ is active, we know that
  \[
    |Y'| \ge t^{-\eps_i}|\cI^i|.
  \]
  Also, by Eq.~(\ref{eq:decreasing-densities-works})
  \[
    |Y| - |Y'| \le 2|\cI^i \setminus \cI^i_{dd}| \le 2t^{-2\eps_i}|\cI^i|.
  \]
  Let $P = t^{-12\eps_i}\rho^{-1}.$  By Lemma~\ref{lem:color}, all but $2t^{-2\eps_i}|Y|$ vertices $y \in Y'$ have degree at least $P|W|$. Thus,
  \[
    \frac{|E_1|}{|Y \times W|} = (1 - 2t^{-2\eps_i})P\frac{|Y'|}{|Y|} \ge (1 - 4t^{-2\eps_i})P.
  \]
  Let $E_2 \subset E_1$, such that $1 - 4t^{-2\eps_i}$ fraction of the $y \in Y$ have the degree exactly $P|W|$ and all other vertices have degree $0$ (including all vertices in ($Y \setminus Y'$). Thus,
  \[
    \frac{|E_2|}{|Y \times W|} = (1 - 4t^{-2\eps_i})P.
  \]
  Let $W' \subset W$ be the set of nodes for which $C^{i+1}(w, \tau)$ is big or
  \begin{align}
    \deg_{E_2}(w, \tau) &\ge \delta_{i+1}|\cI^{i+1}(w, \tau)| = 2t^{\eps_{i+1}}|\cI^{i+1}(w, \tau)|\label{eq:degree-bound}
  \end{align}
  For any such $(w, \tau)\in W'$ for which $C^{i+1}(w, \tau)$ is not big and the $E_2$ degree is large, we then have that 
  \begin{align*}
    |\cI^{i+1}_{good}(w, \tau)| &\ge |\cI^{i+1}(w, \tau) \cap \mu(\cI^{i+1}(w, \tau))| - 2 |\cI^{i+1} \setminus \cI^{i+1}_{dd}(w, \tau)|\\
                                &\ge 2t^{-\eps_{i+1}}|\cI^{i+1}(w, \tau)(w, \tau)| - 2t^{-2\eps_{i+1}}|\cI^{i+1}| &&\text{Eq.~\eqref{eq:degree-bound} and (\ref{eq:decreasing-densities-works})}\\
                                &\ge t^{-\eps_{i+1}}|\cI^{i+1}(w, \tau)|.
  \end{align*}
  Thus, $\cI^{i+1}(w, \tau)$ is active for all $(w, \tau) \in W'$.  Let $E_3 = E_2 \cap (Y \times W')$. Observe that the average degree of $w \in W$ is
  \begin{align}
    (1-4t^{-2\eps_i})P|Y| &= (1 - 4t^{-2\eps_i})t^{-12\eps_i}\rho^{-1}|\cI^i_{good}|\nonumber\\
                          &\ge (1 - 4t^{-2\eps_i})t^{-13\eps_i}\rho^{-1}|\cI^i| &&\text{$\cL^i$ is active}\nonumber\\
                          &\ge (1 - 4t^{-2\eps_i})t^{-14\eps_i}\max_{\substack{(w,\tau) \in W':\\\text{$C^{i+1}(w, \tau)$ isn't big}}}|\cI^{i+1}(w, \tau)|&&\text{(Proposition~\ref{prop:next-I-small})}\nonumber\\
                          &\ge t^{2\eps_i}(\delta_{i+1}\max |\cI^{i+1}(w, \tau)|).\label{eq:markov}
  \end{align}
  Thus, at most $t^{-2\eps_i}$  fraction of the edges are deleted. Therefore, 
  \[
    \frac{|E_3|}{|Y \times W|} \ge (1 - 5t^{-2\eps_i})P.
  \]
  Note that for all $(w, \tau) \in W'$, $C^{i+1}(w, \tau)$ is big or $\cI^{i+1}(w, \tau)$ is active. Let $E_4 \subset E_3$ be the edges $(y, (w, \tau))$ such that
  \[
    y,\mu(y) \in C^{i+1}(w, \tau)\text{ or } y\text{ survives }\cI^{i+1}(w, \tau).
  \]
  By the inductive hypothesis, for every $(w, \tau) \in W'$, the number of edges of $(w, \tau)$ that get deleted from $E_3$ to $E_4$ is at most
  \begin{align*}
    \delta_{i+1} |\cI^{i+1}(w, \tau) \cap \cI^{i+1}(w, \tau)| &\le \delta_{i+1} |\cI^{i+1}(w, \tau)|\\
    &\le  t^{-2\eps_i}P|Y| && (\text{by Eq.~\ref{eq:markov}}).
  \end{align*}
  Therefore,
  \[
    \frac{|E_4|}{|Y \times W|} \ge (1 - 6t^{-2\eps_i})P.
  \]
  Since every vertex of $Y$ has degree at most $P|W|$, by Markov's inequality, all but $12t^{-2\eps_i} \le \delta_i$ of $y \in Y$ have degree at least $P|W|/2$. Let $Y_{surive}$ be the set of such $y$. Then, for each $y \in Y_{surive}$, the probability that a given $(w, \tau)$ succeeds for $y$ and finds $(y, \mu(y))$ is at least $P/2 (1  - e^{-t^{\eps_{i+1}}}) \ge P/4$. This implies that the probability $(w, \tau)$ succeeds for some $y$ is at least
  \[
    1 - \left(1 - P/4\right)^{t^{\eps_{i+1}}\rho} \ge 1 - e^{-4t^{\eps_{i+1}}\rho/P} \ge 1 - e^{-t^{\eps_i}}.
  \]
 \end{proof}

\subsubsection{Proof of Lemma~\ref{lem:main-query}}\label{sub:proof-main-query}

\begin{proof}
  The runtime and query complexity guarantees are proved in Section~\ref{subsec:runtime}.

  To prove the first point that all returned edges have edit distance $\beta_L \Delta_{query}$. Note that for each clique $C^i(x, \tau)$ and $y \in C^i(x, \tau)$ encountered by the algorithm, $\ED(x, y) \le c^{\tau_{\max}}_L\Delta_{query}$. The edit distance between any two windows in the clique is at most $2c^{\tau_{\max}}\Delta_{query} \le \beta_L \Delta_{query}$. Thus, the guarantee is satisfied.

  To prove the second point, recall that we defined $\mu$ to be the edges of the matching with distance at most $\Delta_{query}$. By Lemma~\ref{lem:inductive}, all but $2t^{-\eps_0}$ fraction of $(y, \mu(y))$ can be discovered with probability at least $1 - e^{-t^{\eps_0}}$. Thus by a union bound, with probability at least $1-te^{-t^{\eps_0}} \le 1 - e^{-t^{\eps}}$, all of these can be discovered together. Thus, the number of unfound edges is at most $t(2t^{-\eps_0}) \le t^{1-\eps}$ with high probability.
\end{proof}

\subsection{Efficiency of the query algorithm}\label{subsec:runtime}
For each function, we compute upper bounds on its running time (outside of recursive calls to $\textsc{Main}$) as well as how many times it calls $\textsc{Main}$. When computing run time, the big-$O$ notation obscures all factors independent of $t$.

\begin{proposition}
  $\textsc{DenseInterval}(\cL, C^{i+1}, \cI^i)$ runs in $O(|\cI^i|/\eps)$ time and makes no queries.
\end{proposition}

\begin{proof}
  Iterating over $I \in \cL$ such that $7I \cap \cI^i \neq \emptyset$ can be done in $O(\cI^i)$ by maintaining for each $y \in \cA \cup \cB$ the $O(1)$ list of intervals which intersect it (and marking which ones have been visited to avoid repeats).

  Computing each relative density takes $O(|7I \cap \cI^i|)$ time in the worst case. Note that
  \[
    \sum_{I \in \cL} |7I \cap \cI^i| = O(|\cI^i|/\eps),
  \]
  as each window is in $O(1/\eps)$ intervals of $\cL$.
\end{proof}

The next analysis is sufficiently obvious, we state it without proof.

\begin{proposition}
  $\textsc{SampleClique}(\cI, \Delta, L)$ makes $|\cI|$ queries to $\textsc{Main}(\cdot, L-1)$ and otherwise takes $O(|\cI|)$ running time.
\end{proposition}

We now unroll the Query recursion with the following two inductive claim.

\begin{claim}[Queries to $\textsc{Main}$.]
  For all $i \in \{0, 1, \hdots, i_{\max}\}$, for any node at level $i$ in the algorithm tree, $\textsc{Query}(\cA, \cB, \cL, i, \cI^i, \Delta, L)$ makes at most $\rho^{(i_{\max} - i)(1 + 1/i_{\max}) + 1}$ calls to $\textsc{Main}(\cdot, L-1)$.
\end{claim}

\begin{proof}
  We prove this by (reverse) induction on $i$. The base case of $i = i_{\max}$, then $|\cI^{i_{\max}}|^2$ queries are made to $\textsc{Main}$. By Proposition~\ref{prop:next-I-small}, we have that
 \[
    |\cI^{i_{\max}}| \le t^{1+5i_{\max}\eps}\rho^{-i_{\max}} = t^{5i_{\max}\eps}.
  \]
  Thus, the number of queries is at most $t^{10i_{\max}\eps} \ll \rho$.

  For $i < i_{\max}$, the outer loop of $\textsc{Query}$ is iterated $t^{\eps_i}\rho$ times. Each of these iterations makes $|\cI^i|$ calls to $\textsc{Main}$, and if the clique is small then a call is made to $\textsc{Query}(\cA, \cB, \cL, i+1, \cI^{i+1}, \Delta, L)$. By the induction hypothesis, this recursive call queries $\textsc{Main}$ at most $\rho^{(i_{\max} - i - 1)(1 + 1/i_{\max}) + 1}$ times. Thus, the total number of calls to $\textsc{Main}$ is at most
  \begin{align*}
    t^{\eps_i}\rho(|\cI^i| + \rho^{(i_{\max} - i - 1)(1 + 1/i_{\max}) + 1}) &\le t^{\eps_i}\rho(t^{1 + 5i\eps}\rho^{-i} + \rho^{(i_{\max} - i - 1)(1 + 1/i_{\max}) + 1})&&\text{(Proposition~\ref{prop:next-I-small})}\\
                                                                            &\le t^{\eps_i + 5i\eps}\rho^{i_{\max} - i + 1} + \frac{t^{\eps_i}}{\rho^{1/i_{\max}}} \rho^{(i_{\max} - i)(1 + 1/i_{\max}) + 1}\\
                                                                            &\le \rho^{(i_{\max} - i)(1 + 1/i_{\max}) + 1},
  \end{align*}
  where the last step follows from $t^{\eps_i + 5i\eps} \ll t^{1/i_{\max}^2}= \rho^{1/i_{\max}}$.
\end{proof}

\begin{claim}[Running time.]
  For all $i \in \{0, 1, \hdots, i_{\max}\}$, for any node at level $i$ in the algorithm tree, $\textsc{Query}(\cA, \cB, \cL, i, \cI^i, \Delta, L)$ has a running time of $\rho^{2(i_{\max} - i - 1)(1 + 1/i_{\max}) + 1}$ outside of calls to $\textsc{Main}(\cdot, L-1)$.
\end{claim}

\begin{proof}
    Like the previous claim, we prove this by (reverse) induction on $i$. The base case of $i = i_{\max}$, $O(|\cI^{i_{\max}}|^2)$ runtime outside of $\textsc{Main}$. By Proposition~\ref{prop:next-I-small}, we have that
  \[
    |\cI^{i_{\max}}| \le t^{1+5i_{\max}\eps}\rho^{-i_{\max}} = t^{5i_{\max}\eps}.
  \]
  Thus, the number of queries is at most $O(t^{10i_{\max}\eps}) \ll \rho$.

  For $i < i_{\max}$, the outer loop of $\textsc{Query}$ is iterated $t^{\eps_i}\rho$ times. Each of these iterations, the runtime is $O(|\cI^i|^2)$ outside of calls to $\textsc{Main}$. If the clique is small then a call is made to $\textsc{Query}(\cA, \cB, \cL, i+1, \cI^{i+1}, \Delta, L)$. By the induction hypothesis, this recursive has a runtime at most $\rho^{2(i_{\max} - i - 1)(1 + 1/i_{\max}) + 1}$ times. Thus, the total number of calls to $\textsc{Main}$ is at most
  \begin{align*}
    t^{\eps_i}\rho(|\cI^i|^2 + \rho^{2(i_{\max} - i - 1)(1 + 1/i_{\max}) + 1}) &\le t^{\eps_i}\rho(t^{2 + 10i\eps}\rho^{-2i} + \rho^{2(i_{\max} - i - 1)(1 + 1/i_{\max}) + 1})&&\text{(Proposition~\ref{prop:next-I-small})}\\
                                                                            &\le t^{\eps_i + 10i\eps}\rho^{2(i_{\max} - i) + 1} + \frac{t^{\eps_i}}{\rho^{2/i_{\max}}} \rho^{2(i_{\max} - i)(1 + 1/i_{\max}) + 1}\\
                                                                            &\le \rho^{2(i_{\max} - i)(1 + 1/i_{\max}) + 1},
  \end{align*}
  where the last step follows from $t^{\eps_i + 10i\eps} \ll \rho^{1/i_{\max}} = t^{1/i_{\max}^2}$. 
\end{proof}

Applying the root of the algorithm tree (at level $i = 0$) from the previous two claims, we get the following corollary.

\begin{corollary}
  $\textsc{Query}(\cA, \cB, \cL, 0, \cI, \Delta, L)$ makes $t^{1+3/i_{\max}}$ queries to $\textsc{Main}(\cdot, L-1)$ and otherwise takes $t^{2 + 3/i_{\max}}$ running-time.
\end{corollary}

\bibliographystyle{alpha}
\bibliography{draft}

\end{document}